\documentclass[letterpaper, twocolumns, unpublished, 10pt]{quantumarticle}

\pdfoutput=1

\usepackage[utf8]{inputenc}
\usepackage[english]{babel}
\usepackage[T1]{fontenc}

\usepackage{amsmath,amssymb} \usepackage{amsthm}
\usepackage{physics}
\usepackage{hyperref}

\usepackage{mdframed}
\usepackage{booktabs}
\usepackage{longtable}

\usepackage{graphicx}
\usepackage{xcolor}

\usepackage[numbers, sort&compress]{natbib}

\usepackage{soul}

\usepackage[ruled,vlined, linesnumbered]{algorithm2e}

\newtheorem{theorem}{Theorem}

\newtheorem{corollary}{Corollary}
\newtheorem{lemma}{Lemma}
\newtheorem{proposition}{Proposition}

\newcommand{\introduce}[1]{\textit{#1}}

\newcommand{\stabgens}{\ensuremath{\mathcal{S}}}

\newcommand{\Z}{\mathbb{Z}}
\newcommand{\N}{\mathbb{N}}

\DeclareMathOperator{\cheeger}{h}
\newcommand{\hepsilon}{{\cheeger}_{\varepsilon}}

\DeclareMathOperator{\cnot}{CNOT}
\DeclareMathOperator{\cut}{cut}
\DeclareMathOperator{\depth}{depth}

\DeclareMathOperator{\con}{\text{c}}
\DeclareMathOperator{\im}{Im}
\DeclareMathOperator{\Index}{index}
\newcommand{\CEC}{{{\cal D}_C}}

\title{Bounds on stabilizer measurement circuits and obstructions to local implementations of quantum LDPC codes}

\author{Nicolas Delfosse}
\author{Michael E. Beverland}

\affiliation{
    Microsoft Quantum
    \&
    Microsoft Research,
    Redmond, WA 98052, USA
}

\author{Maxime A. Tremblay}
\affiliation{
    Institut quantique
    \&
    Département de physique,
    Université de Sherbrooke,
    Sherbrooke, Qc, Canada, J1K 2R1
}

\begin{document}

\maketitle

\begin{abstract}
In this work we establish lower bounds on the size of Clifford circuits that measure a family of commuting Pauli operators.
Our bounds depend on the interplay between a pair of graphs: the Tanner graph of the set of measured Pauli operators, and the connectivity graph which represents the qubit connections required to implement the circuit.
For local-expander quantum codes, which are promising for low-overhead quantum error correction, we prove that any syndrome extraction circuit implemented with local Clifford gates in a 2D square patch of $N$ qubits has depth at least $\Omega(n/\sqrt{N})$ where $n$ is the code length.
Then, we propose two families of quantum circuits saturating this bound.
First, we construct 2D local syndrome extraction circuits for quantum LDPC codes with bounded depth using only $O(n^2)$ ancilla qubits.
Second, we design a family of 2D local syndrome extraction circuits for hypergraph product codes using $O(n)$ ancilla qubits with depth $O(\sqrt{n})$.
Finally, we use circuit noise simulations to compare the performance of a family of hypergraph product codes using this last family of 2D syndrome extraction circuits with a syndrome extraction circuit implemented with fully connected qubits.
While there is a threshold of about $10^{-3}$ for a fully connected implementation, we observe no threshold for the 2D local implementation despite simulating error rates of as low as $10^{-6}$. This suggests that quantum LDPC codes are impractical with 2D local quantum hardware.
We believe that our proof technique is of independent interest and could find other applications.
Our bounds on circuit sizes are derived from a lower bound on the amount of correlations between two subsets of qubits of the circuit and an upper bound on the amount of correlations introduced by each circuit gate, which together provide a lower bound on the circuit size.
\end{abstract}

\section{Introduction and overview}
\label{sec:intro}

Quantum Low Density Parity Check (LDPC) codes~\cite{mackay_sparse-graph_2004, tillich_quantum_2014, hastings2021fiber, panteleev2020quantum}, capable of encoding many logical qubits within the same block, offer promising performance for large scale fault-tolerant quantum computation.
Quantum LDPC codes have two appealing features:
(i) they are defined by constraints involving a small number of qubits,
(ii) some can encode $k$ logical qubits into $n$ physical qubits with $k = \Omega(n)$, with a large minimum distance $d$.
Recall that a minimum distance $d$ guarantees that the code can correct any error acting on up to $(d-1)/2$ qubits\footnote{In the ideal case of perfect syndrome extraction.}.
Property (i) makes quantum LDPC codes easier to implement than general codes, and (ii) implies quantum LDPC codes perform well and could lead to a more favorable overhead than the surface code~\cite{gottesman_fault-tolerant_2014, fawzi_constant_2018}.

Numerical simulations of the performance of quantum LDPC codes without geometric locality constraints show encouraging results~\cite{kovalev2018numerical, grospellier_numerical_2019, grospellier_combining_2020, tremblay2021layers}.
In this work, we explore the potential of quantum LDPC codes implemented with quantum hardware limited to local operations in a 2D grid of qubits which matches the capability of many hardware technologies such as superconducting qubits~\cite{arute2019quantum, chamberland2020building} and Majorana qubits~\cite{karzig2017scalable}.

To correct errors affecting an encoded quantum state, we first execute a quantum circuit, the so-called syndrome extraction circuit, which measures the local constraints defining the code. 
Then, the syndrome output by the syndrome extraction circuit is fed to the decoder which is a classical subroutine that processes the syndrome and returns a correction. 
Finally, the encoded state is corrected using the output of the decoder.
 
Decoders have been proposed for quantum LDPC codes~\cite{leverrier_quantum_2015, panteleev2019degenerate, roffe2020decoding, delfosse2021toward} which can be combined with any syndrome extraction circuit. 
With 2D connectivity, we expect the bottleneck for error correction performance to be the syndrome extraction circuit.
However, the question of designing 2D local syndrome extraction circuits for quantum LDPC codes remains open.

Locality is known to restrain the parameters achievable with quantum error correction codes.
Bravyi, Poulin and Terhal~\cite{bravyi_tradeoffs_2010} proved that quantum codes defined by local commuting projectors in a 2D grid of $\sqrt{n} \times \sqrt{n}$ qubits obey the bound $k d^2 \leq O(n)$.
This result translates into a bound on the depth of syndrome extraction circuits built from with local unitary gates and local measurements in a 2D square patch.
Indeed, assume that a family of quantum codes with length $n \rightarrow +\infty$ encoding $k \leq Rn$ logical qubits can be implemented with 2D local syndrome extraction circuits with bounded depth.
The projectors defining a code can be obtained by backpropagating the projectors corresponding to local measurements in the circuit.
For bounded-depth local syndrome extraction circuits, this procedure leads to a family of local commuting projectors in 2D.
If the circuit is supported on a square patch of $n$ qubits, this implies that the parameters of the codes are limited by the tradeoff $k d^2 \leq O(n)$ which leads to $d = O(1)$.
The only quantum codes that can be implemented in bounded depth with these assumptions have a bounded minimum distance and therefore a non-vanishing logical error rate (and no threshold).

\begin{figure}[t]
  \centering
  \includegraphics[scale=.4]{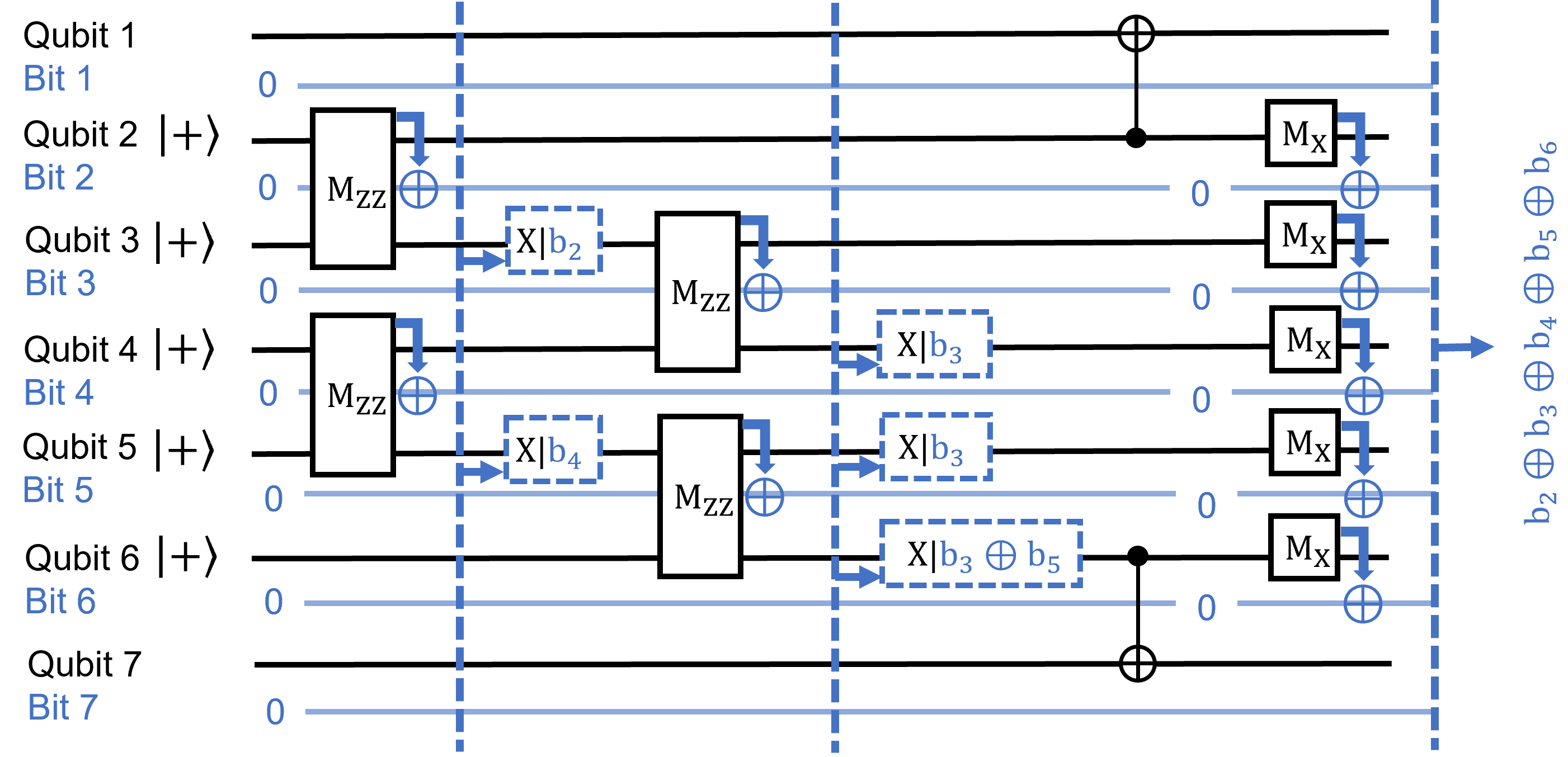}
  \caption{
  A local circuit on a line for qubits implementing the measurement of the Pauli operator $XX$ supported on the endpoints of the line (qubit 1 and qubit 7). 
  Each qubit comes with a classical memory represented by the light blue line which is used to store the measurement outcomes.
  The circuit is made with single-qubit and two-qubit quantum operations, represented in black, acting on adjacent qubits.
  The blue boxes represent classically-controlled Pauli operations that are applied only if a subset of bits has parity $1$.
  These conditional gates may depend on any subset of bits and require long-range classical communication, represented by blue vertical dashed lines to access the classical bit stored.
  }
  \label{fig::long_range_MXX}
\end{figure}

This argument does not preclude bounded-depth syndrome extraction circuits for general quantum LDPC codes with realistic hardware constraints in 2D because it does not allow for classical communication. 
It is reasonable to assume that a quantum chip is equipped with at least some long-range classical communication channels. 
As an example, the circuit of Fig.~\ref{fig::long_range_MXX} can be used to perform a Pauli measurement between the two endpoints of a line of qubits in bounded depth using long-range classical communication.
However, it is unclear if we can simultaneously measure a large set of Pauli operators in bounded depth.

Moreover, if the syndrome extraction cannot be implemented in bounded depth using a total of $O(n)$ qubits, then what is the minimum depth of a 2D local syndrome extraction circuit using $O(n)$ qubits? Alternatively, what is the minimum number of ancilla qubits necessary to implement the syndrome extraction in bounded depth with a 2D local circuit? Finally, how do quantum LDPC codes perform with these syndrome extraction circuits?

We explore all these questions in this article.
We consider a class of 2D local Clifford circuits made with single-qubit and two-qubit unitary Clifford gates and single-qubit and two-qubit Pauli measurements.
We also authorize classically-controlled Pauli operations controlled on the parity of any set of outcomes of previous measurements.
By extending Clifford operations with these classical controlled gates, which exploit long-range classical communication, we obtain a class of Clifford circuit which can implement non-local Pauli measurement in bounded depth through cat states.
We assume no restriction on the classical storage of the measurement outcomes or the classical communication. As a result, the parities controlling different conditional Pauli operations of the same layer of a circuit can be computed instantaneously and simultaneously even if some conditions depend on the same bits or depend on bits that are stored far away from the qubits supporting the Pauli operation.

Our main technical result is a general lower bound on the depth of Clifford circuits implementing the measurement of a family of Pauli operators. This bound relies on the connectivity graph $G_{\con}$ of the circuit which is the graph whose vertices support the qubits and such that two qubits are connected if the circuit contains a two-qubit operation acting on them.

\begin{theorem} \label{theorem:separator_bound}
Let $C$ be a Clifford circuit measuring commuting Pauli operators $S_1, \dots, S_r$.
Then, for any subset of qubits $L$, we have
$$
\depth(C) \geq \frac{n_{\cut}}{64 |\partial L|},
$$
where $n_{\cut}$ is the number of independent operators $S_i$ with support on both $L$ and its complement and $\partial L$ is the set of edges connecting $L$ and its complement in the connectivity graph of the circuit.
\end{theorem}

This result is obtained by studying the correlations created between two subsets of qubits of the circuit. We believe that our proof technique could find other applications to place bounds on quantum circuits and we provide an overview of our proof strategy in Section~\ref{subsec:proof_strategy}.

Applying Theorem~\ref{theorem:separator_bound}, we obtain bounds on the circuit depth and the number of ancilla qubits required to perform the syndrome extraction of quantum LDPC codes in different settings.
For families of local-expander quantum LDPC codes with length $n$, if the syndrome extraction is implemented with a 2D local Clifford circuit acting a $\sqrt{N} \times \sqrt{N}$ patch of qubits, we prove (Corollary~\ref{cor:dense_2d_local_meas_circuit}) that the depth of the syndrome extraction circuits satisfies 
\begin{align} \label{eqn:intro_syndrome_circuit_bound}
\depth(C) \geq \Omega\left(\frac{n}{\sqrt{N}}\right) \cdot
\end{align}
We establish similar results for $D$-dimensional local Clifford circuits (Corollary~\ref{cor:dense_D_dim_local_meas_circuit}) and for 2D local Clifford circuits acting on an arbitrary subset of the square grid $\Z^2$ (Corollary~\ref{cor:general_2d_local_meas_circuit}).
The assumption of local-expansion is necessary because surface codes violate these bounds. However, their parameters are limited by the tradeoff~\cite{bravyi_tradeoffs_2010}.
Many standard families of classical and quantum LDPC codes exhibit local-expansion such as random classical LDPC codes\cite{gallager_low-density_1962}, classical expander codes~\cite{sipser_expander_1996}, quantum hyperbolic codes~\cite{zemor2009cayley, guth2014quantum} and quantum hypergraph product codes~\cite{tillich_quantum_2014}. 
The assumption of local expansion is also commonly used to guarantee that the decoder performs well~\cite{sipser_expander_1996, hastings2013decoding, leverrier_quantum_2015, delfosse2021toward} and some notion of expansion is known to be required to obtain good quantum LDPC codes~\cite{baspin2021connectivity}.

Then, we propose two classes of syndrome extraction circuits saturating the bound~\eqref{eqn:intro_syndrome_circuit_bound}.
On the one hand, we design 2D local Clifford syndrome extraction circuits for any family of CSS codes that run in bounded depth and use only $O(n^{2})$ ancilla qubits which is the minimum number of ancilla qubits for a bounded depth circuits based on~\eqref{eqn:intro_syndrome_circuit_bound}.
On the other hand, we design 2D local Clifford syndrome extraction circuits for any family of hypergraph product codes~\cite{tillich_quantum_2014} that run in depth $O(\sqrt{n})$ and use only $O(n)$ ancilla qubits. Based on the bound~\eqref{eqn:intro_syndrome_circuit_bound}, this is the best achievable depth when using $O(n)$ ancilla qubits.

\begin{figure}[t]
  \centering
  \includegraphics{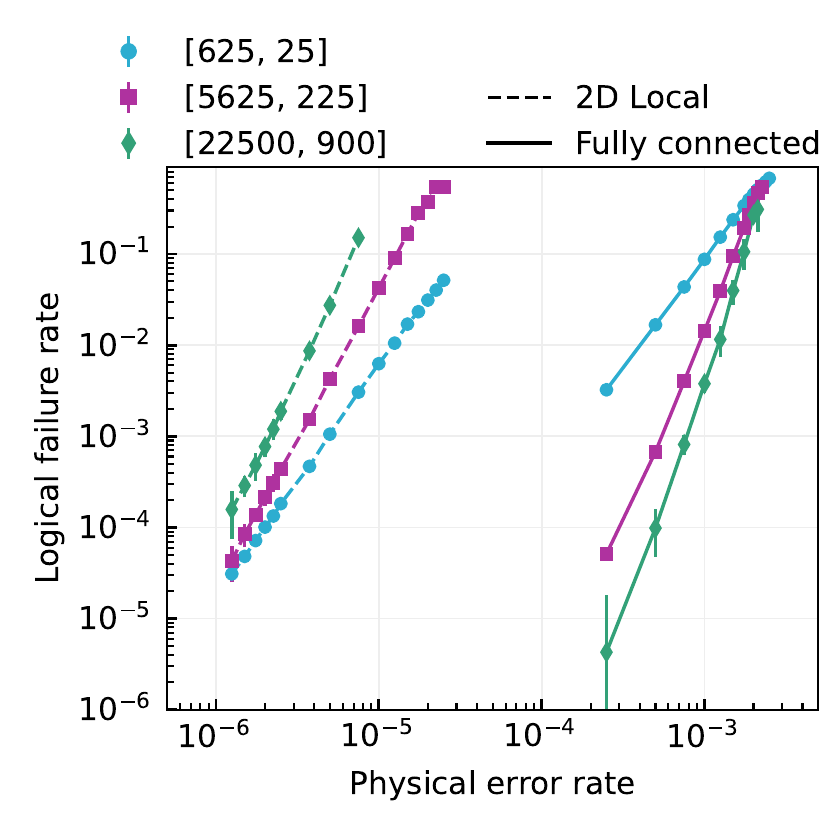}
  \caption{
  Logical failure rate as a function of the gate error rate for a family of hypergraph product codes with rate 1/25, using either the 2D local syndrome extraction circuit described in Section~\ref{sec::constant_overhead_circuits} or a bounded-depth syndrome extraction circuit described in the beginning of Section~\ref{sec::constant_depth_circuits} made with two-qubit gates acting on arbitrary pairs of qubits.
  }
  \label{fig::numerical_results}
\end{figure}

To assess the performance of quantum LDPC codes of finite (rather than asymptotic) length, we perform simulations with good 2D local syndrome extraction circuits.
In particular, we simulate the performance of a family of hypergraph product codes~\cite{tillich_quantum_2014} with encoding rate $k/n = 1/25$. 
To preserve the high rate of logical qubits per physical qubits of these codes, we use our second family of syndrome extraction circuits which uses $O(n)$ ancilla qubits. Including the ancilla qubits used for syndrome extraction, the rate of the family is $k/N = 1/98$.
Fig.~\ref{fig::numerical_results} shows the performance of this family of quantum LDPC codes in two different settings. 
The code and the decoder are the same in both settings but the syndrome extraction is implemented either with our 2D local syndrome extraction circuit or with a bounded-depth syndrome extraction circuit using arbitrary two-qubit gates, ignoring geometric constraints.
This second case assumes fully connected qubits.
Both schemes are simulated with circuit level noise.
In the regime of our simulation, we observe that the 2D local implementation of this family of quantum LDPC codes achieves a performance comparable to the fully connected implementation with a physical noise rate increased by a factor 100.
Moreover, we observe in Fig.~\ref{fig::numerical_results} that increasing the block size degrades the performance in 2D local case which shows that if this family of codes has a non-zero error threshold in 2D, it is below $10^{-6}$.
In a companion paper~\cite{tremblay2021layers}, we explore an alternative implementation of these quantum LDPC codes using some long-range connections between the qubits.

We introduce definitions and background material in Section~\ref{sec:background}.
Our bounds on the depth and the number of ancilla qubits of local syndrome extraction circuits are proven in Section~\ref{sec:circuit_bound}.
The proof of our main technical result, Theorem~\ref{theorem:separator_bound}, is provided in Section~\ref{sec:proofs}.
Section~\ref{sec::constant_depth_circuits} describes a construction of bounded-depth 2D local syndrome extraction circuits for CSS codes.
The 2D local syndrome extraction circuits for hypergraph product codes using a linear number of ancilla qubits are proposed in Section~\ref{sec::constant_overhead_circuits}.
The last section, Section~\ref{sec::numerical_results}, describes our numerical results.

\section{Background and definitions}
\label{sec:background}

Here we review some background topics and provide definitions that will be used throughout the rest of the paper.

\subsection{Clifford circuits}
\label{subsec:clifford_circuits}

A quantum circuit is a sequence of operations acting on a set of qubits.
There are many important classes of circuit, which can typically be defined by restricting the set of allowed operations from which the circuit is composed.
In this work, we focus on {\em Clifford circuits} which are built from the following operations.
\begin{itemize}
\item Preparations of $\ket 0$ or $\ket +$.
\item Single-qubit and two-qubit Pauli measurements.
\item Single-qubit and two-qubit unitary Clifford gates.
\item Classically-controlled Pauli operations, applied only if some subset of previous measurement outcomes has parity 1.
\item Output a set of classical bits obtained by computing the parity of some subsets of measurement outcomes.
\end{itemize}
To avoid any confusion with the outcomes of the single-qubit and two-qubit measurements of the circuit, we reserve the term {\em output} to refer to the parity bits returned at the end of the circuit.
We assume that the Pauli operation has access to all the measurement outcomes extracted in the past.
We discuss the generalization of our results to circuits involving more general operations in Section~\ref{sec:generalizations}.

A Clifford circuit is decomposed into {\em layers} of the above operations in such a way that the support of two operations of the same layer do not overlap and the {\em depth} of the circuit is the number of layers.
Here we define the support of a classically-controlled Pauli operation to be the set of qubits acted on by the Pauli operation. 
The support of a classically-controlled Pauli operation is independent of the classical bits controlling the gate.
Therefore, the conditions of two classically-controlled Pauli operations of the same layer may overlap.
We assume unlimited classical communication and classical storage capacity and do not account for any delay due to classical control.
For example, Fig.~\ref{fig::long_range_MXX} shows a circuit with depth six.
Here, we include Pauli operations in the calculation of the depth of the circuit.
Given that Pauli operations can be implemented by frame update, which requires only classical processing, we could also ignore Pauli operations in the calculation of the depth of quantum circuits. 
With this notion of depth, our lower bounds on the depth of a quantum circuit would remain unchanged up to a constant factor.

Given a circuit $C$, we define the \emph{connectivity graph of the circuit} $G_{\con} = G_{\con}(C)$ which has a vertex for each qubit, and an edge connecting each pair of vertices that are in the support of any single operation applied in the circuit.

On the other hand, consider a set of qubits whose connectivity is described by a graph $G$, {\em i.e.} the two-qubit gates available are the gates supported on the edges of the graph $G$. A circuit $C$ is said to be implementable with qubits with connectivity $G$ if $G_{\con}(C)$ is a subgraph of $G$.

A {\em $D$-dimensional $b$-local circuit} is a circuit acting on qubits placed on a subset of the grid $\Z^D$ such that each gate acts on qubits at distance at most $b$ from each other, where distances are calculated with the infinity norm.
We sometimes refer to a family of $D$-dimensional $b$-local circuits as a {\em family of $D$-dimensional local circuits}, omitting the constant $b$.
For example, measuring the stabilizer generators of the surface code can be implemented by a two-dimensional $1$-local circuit family of bounded depth using a square-grid connectivity graph~\cite{fowler_surface_2012}.

\subsection{Pauli measurement circuits}

Consider a set of commuting $n$-qubit independent Pauli operators $\mathcal{S}=\{S_1, \dots, S_r \}$.
For $m = (m_1, \dots, m_r) \in \Z_2^r$, denote by $\Pi_m$ the projector onto the common eigenspace of the $S_i$ with eigenvalue $(-1)^{m_i}$.
Following the postulates of quantum mechanics, a {\em Pauli measurement circuit} for the measurement of $\mathcal{S}$ is a Clifford circuit which takes as an input a $n$-qubit state with density matrix $\rho$ and returns the output $m = (m_1, \dots, m_r) \in \Z_2^r$ with probability 
$
\Tr(\Pi_m \rho).
$
Moreover, the state of the $n$ input qubits must be mapped onto 
$
\Pi_m \rho \Pi_m / \Tr(\Pi_m \rho).
$ 

We call the $n$ qubits supporting the operators in $\mathcal{S}$ the {\it data qubits},
and any number $a_n$ of other qubits in the circuit the {\it ancilla qubits} such that the total number of qubits in the circuit is $N = n + a_n$.
We assume that ancilla qubits are prepared in the state $\ket 0$ or $\ket +$ before being used and they are returned to a state $\ket 0$ or $\ket +$ at the end of the circuit.
We also assume that data qubits are never measured with a single qubit measurement.

Like in Section~\ref{subsec:clifford_circuits}, we assume that the output $m_i$, corresponding to the measurement of $S_i$, is obtained by taking the parity of a subset $O_i$ of outcomes of single-qubit and two qubit measurements of the circuit, that is
$
m_i = \oplus_{o \in O_i} o,
$
where $\oplus$ denotes the sum modulo 2.
These outcomes $o$ may be extracted from any layer of the circuit and the sets $O_i$ may overlap.

\subsection{Tanner graph and contracted Tanner graph}

Consider a set $\stabgens=\{S_1, \dots, S_r \}$ of $n$-qubit Pauli operators acting on qubits that we denote $q_1, \dots, q_n$.
The \introduce{Tanner graph} of the set of Pauli operators $\stabgens$, denoted $T(\stabgens)$ or simply $T$ when no confusion is possible, is defined to be the bipartite graph $T = (V_Q \cup V_{\stabgens}, E)$ where $V_Q = \qty{q_1, q_2,\ldots, q_n}$ and $V_{\stabgens} = \qty{S_1, S_2,\ldots, S_r}$ correspond respectively to the data qubits and the Pauli operators.
The vertices $q_i$ and $S_j$ are connected by an edge if and only if the operator $S_j$ acts non-trivially on qubit $q_i$.

The {\em contracted Tanner graph} $\bar T(\stabgens)$ or simply $\bar T$, is the graph with vertex set $V_Q$ such that two vertices $q_i$ and $q_j$ are connected by an edge if and only if there exists an operator $S_k$ that acts non-trivially on both $q_i$ and $q_j$.
Alternatively, the contracted Tanner graph can be obtained by connecting two data qubits if and only if they are at distance two in the Tanner graph.

\subsection{Stabilizer codes and quantum LDPC codes}

A \introduce{stabilizer code} with length $n$ is the common $+1$-eigenspace of a set of $n$-qubit independent commuting Pauli operators $\stabgens = \qty{S_{1}, S_{2}, \ldots, S_r}$.
We refer to the operators $S_i$ as the {\em stabilizer generators} of the code.
When considering a stabilizer code, we always assume that it comes with fixed set of stabilizer generators.
By the Tanner graph of a stabilizer code, we mean the Tanner graph of its set of stabilizer generators.

By a {\em syndrome extraction circuit} for the stabilizer code with stabilizer generators $\mathcal{S}$, we mean a Pauli measurement circuit for the set $\mathcal{S}$.

A {\em family of quantum LDPC codes} is a family of stabilizer codes $(Q_i)_{i \in \N}$ such that the Tanner graph $T_i$ of $Q_i$ has degree bounded by some constant $c$ independent of $i$.
Equivalently, there exists a constant $c$ such that for all $i \in \N$, each stabilizer generator has weight at most $c$ and each qubit is acted on by at most $c$ stabilizer generators.

In Sections~\ref{sec::constant_depth_circuits} and \ref{sec::constant_overhead_circuits}, we focus on Calderbank Shor Steane (CSS) codes~\cite{calderbank_good_1996, steane_simple_1996}.
These are codes for which each stabilizer generator is either composed entirely of
$I$ and $X$ operators or entirely of $I$ and $Z$ operators, such that $\stabgens = \stabgens_X \cup \stabgens_Z$ where $\stabgens_X$ is the set of $X$-type stabilizer generators etc.
In the case of CSS codes, we define the $X$ Tanner graph $T_X$ as the subgraph of $T$ induced by the vertices corresponding to the qubits and the $X$ stabilizers.
We define the $Z$ Tanner graph $T_Z$ similarly.

\subsection{Expander graphs and expander codes}

We introduce the following local generalization of the {\em Cheeger constant} of a graph $G = (V, E)$, denoted $\hepsilon$, defined as
\begin{align*}
\hepsilon(G) = \min_{\substack{L \subseteq V \\ |L| \leq \varepsilon|V|/2}} \frac{|\partial L|}{|L|},
\end{align*}
for all $\varepsilon \in [0, 1]$.
Therein, $\partial L$ is the boundary of $L$, that is the set of edges connecting $L$ and its complement. 
The usual Cheeger constant, denoted $h(G)$ is obtained by setting $\varepsilon = 1$, that is $h(G) = h_1(G)$.

A {\em family of $\alpha$-expander graphs} is a family of graphs $(G_i)_{i \in \N}$ such that $h(G_i) \geq \alpha$ for all $i \in \N$. 
We consider a generalization of this notion by considering expansion over small subsets of vertices.
A {\em family of $(\alpha, \varepsilon)$-expander graphs} is a family of graphs $(G_i)_{i \in \N}$ such that $h_{\varepsilon}(G_i) \geq \alpha$ for all $i \in \N$.

Clearly, if $(G_i)_{i \in \N}$ is a family of $\alpha$-expander graphs, it is also a family of $(\alpha, \varepsilon)$-expander graphs for all $\varepsilon \in [0, 1]$. However, $(\alpha, \varepsilon)$-expansion with $\varepsilon < 1$ does not imply $\alpha$-expansion.

By a {\em family of quantum expander codes} we mean a family of quantum LDPC codes such that the family of contracted Tanner graphs of the stabilizer generators is a family of $\alpha$-expander graphs for some $\alpha > 0$.
Similarly, we define a {\em family of local-expander codes} as a family of quantum LDPC codes equipped with $(\alpha, \varepsilon)$-expander contracted Tanner graphs for some $\alpha, \varepsilon > 0$.
By definition, a family of expander codes is also a family of local-expander codes.

It is common in the literature to define quantum expander codes or local quantum expander codes in terms of their Tanner graphs rather than their contracted Tanner graphs.
In Lemma~\ref{lemma:Tanner_expander_implies_meas_expander}, we show that if a code's Tanner graph is a local expander graph, then so too is its contracted Tanner graph. 
This allows our results to be applied to code families with known expansion properties of the Tanner graphs.

\begin{lemma} \label{lemma:Tanner_expander_implies_meas_expander}
Let $T$ be the Tanner graph of a stabilizer code with length $n$ and with $r$ stabilizer generators and let $\bar T$ be its contracted Tanner graph.
Then, for all $\varepsilon \in [0, 1]$, we have
\begin{align*}
\text{h}_{\varepsilon'}(\bar T) \geq \frac{\hepsilon(T)}{\deg(T)},
\end{align*}
where $\varepsilon' =  \frac{n+r}{(1 + \deg(T)) n} \varepsilon$.
\end{lemma}

\begin{proof}
A subset $L \subseteq V(\bar T)$ can be treated as a subset of $V(T)$.
To avoid any confusion, we use the notation $\partial_{\bar T} L$ and $\partial_T L$ for the set of edges leaving $L$ in the two graphs.

To enumerate edges of $\partial_{\bar T} L$, we introduce a function $\varphi_L: \partial_{\bar T} L \rightarrow {\mathcal P}(\partial_T (L \cup N_T(L)))$ where ${\mathcal P}(A)$ denotes the power set of $A$.
The function $\varphi_L$ maps the edge $e = \{u, v\}$ with $u \in L$ and $v \notin L$ onto the set of edges $\{w, v\}$ of $T$ such that $w$ is adjacent to both $u$ and $v$ in $T$.
On the one hand, there are at most $\deg(T)$ edges in the image of an edge $e \in \partial_{\bar T} L$ and each edge of $\partial_T (L \cup N_T(L))$ belongs to one of the subsets of the image of $\varphi_L$.
This shows that $|\im \varphi_L| \geq |\partial_T (L \cup N_T(L))| / \deg(T)$.
On the other hand, we have $|\im \varphi_L| \leq |\partial_{G_m} L|$.
This proves that
\begin{align} \label{eq:lemma_expansion_proof}
|\partial_{\bar T} L| \geq \frac{|\partial_T (L \cup N_T(L))|}{\deg(T)} \cdot
\end{align}

Let $\varepsilon \in [0, 1]$ and let $\varepsilon' = \frac{n+r}{(1 + \deg(T)) n} \varepsilon$.
Consider a subset $L$ of $V(\bar T)$ such that 
$|L| \leq \varepsilon' |V(\bar T)| / 2$.
Then, we have 
\begin{align}
|L \cup N_T(L)|
& \leq (1 + \deg(T)) |L| \\
& \leq (1 + \deg(T)) \varepsilon' \frac{|V(\bar T)|}{2} \\
& = \varepsilon \frac{|V(T)|}{2} \cdot
\end{align}
Therein, we used $|V(T)| = n+r$ and $|V(\bar T)| = n$.

Given that $|L \cup N_T(L)| \leq \varepsilon |V(T)| / 2$, we can use the Cheeger constant of $T$, which yields
\begin{align}
|\partial(L \cup N_T(L))| 
\geq \hepsilon(T) |L \cup N_T(L))|
\geq \hepsilon(T) |L| \cdot
\end{align}
Combining this with Eq.~\eqref{eq:lemma_expansion_proof}, we get
\begin{align}
|\partial_{\bar T} L| 
& \geq \frac{|\partial_T (L \cup N_T(L))|}{\deg(T)} \\
& \geq \frac{\hepsilon(T) |L|}{\deg(T)},
\end{align}
proving the lemma.
\end{proof}

\section{Circuit bounds}
\label{sec:circuit_bound}

Our main technical result (Theorem~\ref{theorem:separator_bound}) establishes a lower bound on the depth of Clifford circuits measuring a set of commuting Pauli operators $S_1, \dots, S_r$. 
In what follows, $\bar T$ denotes the contracted Tanner graph of the set of measured operators $S_1, \dots, S_r$.
In this section, we provide different applications of Theorem~\ref{theorem:separator_bound}.
The proof of this theorem is deferred to Section~\ref{sec:proofs}.

\subsection{Local circuits on qubits on a patch of $\Z^2$}

In what follows, we use the terminology of an $N$-qubit {\em patch} of $\Z^2$ to refer to a $\sqrt{N} \times \sqrt{N}$ square grid of qubits.

\begin{proposition} \label{prop:dense_2d_local_meas_circuit}
Let $C$ be a 2D $b$-local Clifford Pauli measurement circuit supported on an $N$-qubit patch of $\Z^2$.
Then, for all $\varepsilon \in [0, 1]$ we have
$$
\depth(C) \geq c \frac{\hepsilon(\bar T)}{b^3w(w-1)} \frac{\varepsilon n/2 - \sqrt{N}}{\sqrt{N}},
$$
for some constant $c$, where $\bar T$ and $w$ are respectively the contracted Tanner graph and maximum weight of the set of measured Pauli operators.
\end{proposition}

\begin{proof}
Consider the patch with integer lattice coordinates $(x,y)$. 
For an integer $i$, the subset of qubits supported on the vertical line $x=i$ is denoted $Q_{x=i}$ and $Q_{x \leq u}$ is the union of all the sets $Q_{x=i}$ with $i \leq u$.

Let ${\bf D}$ be the set of data qubits of the circuit. By shifting $u$, we can find a position such that
\begin{align*}
\varepsilon n/2 - \sqrt{N}
\leq
|{\bf D} \cap Q_{x \leq u}|
\leq \varepsilon n/2.
\end{align*}
This value $u$ exists because when $u$ moves by one unit at most $\sqrt{N}$ data qubits change side.

Denote by $c_b = (2b+1)^2$ the number of vertices contained in a closed ball with radius $b$ in $\Z^2$.
Using the $b$-locality of the circuit, we see that the set $L$ of all the qubits included in $Q_{x \leq u}$ satisfies
\begin{align}
|\partial L|
\leq \sum_{i=u-b+1}^{u}  c_b |Q_{x=i}|
\leq c b^3 \sqrt{N},
\end{align}
for some constant $c$.

Moreover, using the Cheeger constant of the contracted Tanner graph, the number of operators $S_i$ with support on both $L$ and its complement is at least 
\begin{align}
n_{\cut}
& \geq \frac{2\hepsilon(\bar T)}{w(w-1)} |{\bf D} \cap L| \\
& \geq \frac{2\hepsilon(\bar T)}{w(w-1)} (\varepsilon n/2 - \sqrt{N}) \cdot
\end{align}
The term $2/(w(w-1))$ is present because each $S_i$ induces at most $\binom{w}{2}$ edges in the contracted Tanner graph.

Applying Theorem~\ref{theorem:separator_bound} with this set $L = Q_{x \leq u}$, we get
\begin{align}
\depth(C) 
\geq c' \frac{\hepsilon(\bar T)}{b^3w(w-1)} \frac{\varepsilon n/2 - \sqrt{N}}{\sqrt{N}},
\end{align}
for some constant $c'$.
\end{proof}

The following corollary is a straightforward application of Proposition~\ref{prop:dense_2d_local_meas_circuit}.

\begin{corollary} \label{cor:dense_2d_local_meas_circuit}
Let $(C_i)_{i \in \N}$ be a family of 2D local Clifford syndrome extraction circuits for a family of local-expander quantum LDPC codes with length $n_i \rightarrow \infty$ where $C_i$ acts on qubits in a patch of $\Z^2$.
Then, we have
\begin{align*}
\depth(C_i) \geq \Omega\left(\frac{n_i}{\sqrt{N_i}}\right),
\end{align*}
where $N_i$ is the total number of qubits used by the circuit. 
\end{corollary}

For simplicity, we have considered here a circuit with $N$ qubits filling a $\sqrt{N} \times \sqrt{N}$ square grid. 
However Corollary~\ref{cor:dense_2d_local_meas_circuit} also holds for a family of circuits acting on a subset of qubits occupying a constant fraction of this $\sqrt{N} \times \sqrt{N}$ square grid because this is equivalent to replacing $N$ by $N' = \Omega(N)$.
Therefore, this corollary holds for qubits placed on a hexagonal or a triangular lattice.
The proof of Proposition~\ref{prop:dense_2d_local_meas_circuit} and the statement of  Corollary~\ref{cor:dense_2d_local_meas_circuit} can be readily modified to apply to rectangular patches $\ell \times \ell'$, where $\ell$ and $\ell'$ can grow at different paces in the code family.

\subsection{Local circuits on qubits on a patch of $\Z^D$}

By an $N$-qubit {\em patch} of $\Z^D$ we mean a $D$-dimensional cubic grid of qubits where the length of each side of a $D$-cube is $N^{1/D}$.
Proposition~\ref{prop:dense_2d_local_meas_circuit} immediately generalizes to $D$-dimensions as follows.

\begin{proposition} \label{prop:dense_D_dim_local_meas_circuit}
Let $D \geq 2$.
Let $C$ be a $D$-dimensional $b$-local Clifford Pauli measurement circuit acting on an $N$-qubit patch of $\Z^D$.
Then, for all $\varepsilon \in [0, 1]$ we have
$$
\depth(C) \geq c \frac{\hepsilon(\bar T)}{b^{D+1}w(w-1)} \frac{\varepsilon n/2 - N^{(D-1)/D}}{N^{(D-1)/D}} ,
$$
for some constant $c$, where $\bar T$ and $w$ are respectively the contracted Tanner graph and maximum weight of the set of measured Pauli operators.
\end{proposition}

We now state the $D$-dimensional analog of Corollary~\ref{cor:dense_2d_local_meas_circuit} which is an immediate application of this proposition.

\begin{corollary} \label{cor:dense_D_dim_local_meas_circuit}
Let $D \geq 2$.
Let $(C_i)_{i \in \N}$ be a family of $D$-dimensional local Clifford syndrome extraction circuits where $C_i$ acts on qubits in a patch of $\Z^D$ for a family of local-expander quantum LDPC codes with length $n_i \rightarrow \infty$.
Then, we have
\begin{align*}
\depth(C_i) \geq \Omega\left(\frac{n_i}{N_i^{(D-1)/D}}\right),
\end{align*}
where $N_i$ is the total number of qubits used by the circuit. 
\end{corollary}

\subsection{Local circuits on qubits on any subset of $\Z^2$}

Here we consider circuits supported on an arbitrary subset of the square grid $\Z^2$.

\begin{proposition} \label{prop:general_2d_local_meas_circuit}
Let $C$ be a 2D $b$-local Clifford Pauli measurement circuit using a total of $N$ qubits on any subset of $\Z^2$.
Then, for all $\varepsilon \in [0, 1]$ we have
$$
\depth(C) \geq c \frac{\hepsilon(\bar T)}{w(w-1)b^6} \frac{\varepsilon^{3/2} n^{3/2}}{N},
$$
for some constant $c$, where $\bar T$ and $w$ are respectively the contracted Tanner graph and maximum weight of the set of measured Pauli operators.
\end{proposition}

\begin{proof}
To construct the subset $L$ of Proposition~\ref{theorem:separator_bound}, we use a separation theorem~\cite{lipton1979separator}. 
We use the variant of the separation theorem proposed in Theorem~3 of \cite{lipton1980applications} which provides a decomposition of a planar graph $(V, E)$ into connected components $(V_i, E_i)$ with size $|V_i| \leq \alpha |V|$ by removing at most $c\sqrt{|V|/\alpha}$ vertices.

Denote by $c_b=(2b+1)^2$ the size of a closed ball with radius $b$. Recall that we use the infinity distance.
We apply the aforementioned separation theorem with $\alpha = \frac{\varepsilon n}{4 c_b N}$, to the subgraph $H_b$ of the square grid containing all the nodes at distance less or equal to $b$ from any of the sites which contain a qubit.
Let $V_1, \dots, V_m$ be the vertex sets of the subgraphs of the decomposition and let $V_0$ be the set of removed vertices.
For $i=1, \dots, m$, we have 
\begin{align} \label{eqn:planar_separation_subset_size}
|V_i| \leq \alpha |V(H_b)| \leq \frac{\varepsilon n}{4},
\end{align}
because the graph $H_b$ contains at most $c_b N$ vertices.

Denote $V_i' = \cup_{j=1}^i V_j$ and let ${\bf D} \cap V_i'$ be the set of data qubits of the circuit included in $V_i'$.
Select the minimum index $i_0$ such that $|{\bf D} \cap V_{i_0+1}'| > \varepsilon n/2$.
Then, we have
\begin{align} \label{eqn:data_qubits_in_L}
\varepsilon n/4
\leq
|{\bf D}  \cap V_{i_0}'|
\leq \varepsilon n/2 \cdot
\end{align}
The upper bound is clear by the definition of $i_0$.
The lower bound is due to the fact that discarding $V_{i_0+1}$ removes at most $\varepsilon n/4$ data qubits based on Eq.~\eqref{eqn:planar_separation_subset_size}.

Let $L$ be the set of all the qubits included in $V_{i_0}'$.
Let us derive an upper bound on the size of $\partial L$ in the connectivity graph.
If $\{u, v\} \in  \partial L$ with $u \in L$ and $v \notin L$, there exists a path with length $\leq b$ joining $u$ and $v$ in the square grid and this path contains at least one vertex $w$ of the removed set $V_0$.
This induces a map from $\partial L$ to $V_0$ and each vertex of $V_0$ has at most $c_b^2$ preimages $\{u, v\}$ because $u$ and $v$ are at distance at most $b$ from $w$.
Therefore, we have
\begin{align}
|\partial L|
& \leq c_b^2 |V_0| \\
& \leq c c_b^2 \sqrt{\frac{|V(H_b)|}{\alpha}} \\
& \leq c c_b^2 \sqrt{\frac{4 c_b^2 N^2}{\varepsilon n}} \\
& \leq 2 c c_b^3 \frac{N}{\sqrt{\varepsilon n}}, \label{eq:upper_bound_partialL}
\end{align}
for some constant $c$.

Moreover, like in the proof of Proposition~\ref{prop:dense_2d_local_meas_circuit}, using the Cheeger constant of the contracted Tanner graph, we get 
\begin{align}
n_{\cut}(L) \geq \frac{2}{w(w-1)} \hepsilon(\bar T) |{\bf D} \cap L|
\geq \frac{\hepsilon(\bar T)}{2w(w-1)} \varepsilon n \cdot \label{eq:lower_bound_ncut}
\end{align}

We conclude by applying Theorem~\ref{theorem:separator_bound} with this set $L$ using Eq.~\eqref{eq:upper_bound_partialL} and Eq.~\eqref{eq:lower_bound_ncut} which yields
\begin{align}
\depth(C) 
& \geq c' \frac{\varepsilon^{3/2} \hepsilon(\bar T)}{w(w-1)b^6} \frac{n^{3/2}}{N} \cdot
\end{align}
\end{proof}

Applying this result to syndrome extraction circuits for families of local-expander quantum LDPC codes, we obtain the following result.

\begin{corollary} \label{cor:general_2d_local_meas_circuit}
Let $(C_i)_{i \in \N}$ be a family of 2D local Clifford syndrome extraction circuits for a family of local-expander quantum LDPC codes with length $n_i \rightarrow \infty$ acting on qubits on any subset of $\Z^2$.
Then, we have
\begin{align*}
\depth(C_i) \geq \Omega\left(\frac{n_i^{3/2}}{N_i}\right),
\end{align*}
where $N_i$ is the total number of qubits used by the circuit.
\end{corollary}

The bound obtained in this section for local circuits on a general subset of $\mathbb{Z}^2$ is weaker than the bound obtained for local circuits on a patch of $\mathbb{Z}^2$ in Proposition~\ref{prop:dense_2d_local_meas_circuit}.
We do not know if this particular bound is tight -- i.e. if there is a smaller-depth syndrome using 2D local extraction circuit acting on qubits placed on a subset of $\mathbb{Z}^2$. 
However, we note that bounded-depth circuits using a linear number of ancilla qubits are still excluded.

\subsection{Generalizations}
\label{sec:generalizations}

Our results can be immediately generalized to other classes of quantum circuits.

First, above we assumed that the Pauli operators measured in a Pauli measurement circuit are independent, but our results directly apply to the measurement of non-independent operators by considering an independent subset.

For simplicity, we assume that circuits involve operations acting on at most two qubits. We can generalize our results to Clifford circuits including Clifford gates and Pauli measurements supported on up to $w$ qubits.
In this case, the connectivity graph of the circuit is a hypergraph defined in such a way that each gate is supported on a hyperedge and the bound of Theorem~\ref{theorem:separator_bound} becomes
$$
\depth(C) \geq \frac{n_{\cut}}{32c_e |\partial L|},
$$
where $c_e = 2 \lfloor w/2 \rfloor$.
Therein, $\partial L$ denotes the set of hyperedges containing at least one vertex in each set $L$ and $L^{\mathsf{c}}$.
The value of the constant $c_e$ corresponds to the maximum increase of entanglement entropy induced by a $w$-qubit unitary gate (Proposition~2 in~\cite{marien2016entanglement}) which is central in the proof of Proposition~\ref{prop:qmi_rate}.

The results of this article also hold for some non-Clifford circuits. Using the exact same proof technique, one can extend our results to the case of quantum circuits made with preparations of $\ket 0$ or $\ket +$, unitary gates acting on up to $w$ qubits, and single-qubit measurements in the computational basis.
However, in this class of circuits, we do not allow classical communication and classically-controlled Pauli operations.
We leave the case of unitary circuits with unconstrained classical communication open.

\section{Proof of Theorem~\ref{theorem:separator_bound}}
\label{sec:proofs}

This section provides the proof of Theorem~\ref{theorem:separator_bound}.

\subsection{Proof strategy}
\label{subsec:proof_strategy}

In this subsection, we sketch the overall strategy for our proof before providing and proving the formal Lemmas which comprise the rigorous proof in the following sections.

We consider a Pauli measurement circuit $C$ which returns the outcomes of the measurement of a set of Pauli operators $S_1, \dots, S_r$.
Our strategy to bound the depth of $C$ is to study correlations introduced by the circuit across a partition of the qubits into two subsets, $L$ and $R$.
On the one hand, because the circuit implements a non-trivial operation, the two sides of the partition are generally correlated at the end of the circuit.
We can use this to derive a lower bound on the amount of correlation created by the circuit.
On the other hand, building these correlations with local gates takes time because we expect that each gate introduces a bounded amount of correlation.
We can use this to derive an upper bound on the amount of correlation created by a circuit as a function of its depth.
Combining both arguments, we obtain a lower bound on the depth of the circuit.

To apply this strategy, we need a protocol which uses the circuit $C$ and a measure of correlation that allows us to readily quantify correlations to build both the lower bound and the upper bound.
In what follows we consider a number of relevant aspects which lead us to such a protocol and measure.

{\bf Initial state.}
To avoid the presence of correlation in the initial state of the system, we fix the input state to be $\ket 0^{\otimes n}$ for the $n$ data qubits.

{\bf Measure of correlation.}
Different notions can be used to capture correlations between the two parts of the partition $L \cup R$ such as the classical mutual information, the entanglement entropy or the quantum mutual information.
Our starting point is to study the classical mutual information between the measurement outcomes extracted in each side of the partition. 
We will refine this notion throughout this section.

{\bf Repeated circuit.}
Consider as an example the Pauli measurement circuit of Fig.~\ref{fig::long_range_MXX} which measures the operator $XX$ supported on the endpoints of a line of qubits, with all qubits initially in the $\ket{0}$ state.
Pick the partition $L \cup R$ of the qubits where $L$ contains the first four qubits and $R$ contains the three remaining qubits.
The output of the circuit is $o = b_2 \oplus b_3 \oplus b_4 \oplus b_5 \oplus b_6$.
Since the input state is fixed to $\ket 0^2$, the value of $o$ will be a uniform random bit, and there is therefore no correlation between the outcomes in $L$ and $R$. 
However, if we run the measurement circuit a second time, we obtain the same output $o = b_2' \oplus b_3' \oplus b_4' \oplus b_5' \oplus b_6'$, even though each individual measurement outcome $b_i'$ is a uniform random bit independent of the previous outcomes $b_j$.
Then, because the circuit output is fixed to $o$, there is one bit of classical mutual information between the outcomes observed on each side during the second run, that is
\begin{align*}
I(b_2', b_3', b_4'; b_5', b_6') = 1 \cdot
\end{align*}
This example encourages us to consider the circuit $C \circ C$ and to use the notion of classical mutual information between measurement outcomes across the two sides of the the partition to capture correlations
\begin{align*}
I(O_L^{(2)}; O_R^{(2)}),
\end{align*}
where $O_L^{(2)}$ (respectively $O_R^{(2)}$) denotes the set of all measurement outcomes extracted on the qubits of $L$ (respectively $R$) during the second run of $C$.

{\bf Highlighting correlations with errors:}
The example considered above allowed us to track the build up of correlations because the output of the simple measurement circuit depends on some measurement outcomes $b_i$ extracted on each side of the partition.
However, if a Pauli operator is measured using a single ancilla qubit, or a set of ancillas that are all supported on the same side of the partition, the argument breaks down.
To highlight the correlations present in the system and that cannot be detected by looking only at the outcomes, we introduce a layer consisting of a random Pauli error $E$ acting on the data qubits between the two runs of $C$ and we consider the circuit $C \circ E \circ C$.
If a measured operator $S_i$ is supported on both sides of the partition, a Pauli error acting on one side may flip the outcome observed on the other side.
Therefore, we measure the correlations using the mutual information
\begin{align*}
I(O_L^{(2)}, E_L; O_R^{(2)}, E_R) \cdot
\end{align*}
where $E_L$ and $E_R$ are the restrictions of the Pauli error $E$ to each subset of qubits.

{\bf Discounting correlations due to classical communication:}
We have argued that we can derive a lower bound on the mutual information 
$I(O_L^{(2)}, E_L; O_R^{(2)}, E_R)$ between the two sides of a partition of the qubits based on the stabilizer generators which are supported on both sides of the partition. 
To apply our strategy we also need an upper bound on the amount of correlation as a function of the circuit depth.
Unfortunately, the quantity $I(O_L^{(2)}, E_L; O_R^{(2)}, E_R)$ is not an appropriate measure of correlation to derive a non-trivial upper bound.
This is because the Clifford circuits that we consider include classical communication through classically-controlled Pauli operations and these operations can boost the mutual information as the illustrated by the following example.

One can create $n$ bits of mutual information in bounded depth for a partition $L \cup R$ of $2n$ qubits using only single qubit operations and classical communication as follows.
Group $2n$ qubits into $n$ pairs, with each pair containing one qubit of $L$ in the state $\ket +$ and one qubit of $R$ in the state $\ket 0$.
First, we measure all the $\ket +$ states in the computational basis and if the outcome is non-trivial we apply an $X$ gate to the other qubit of the pair.
Then, we measure all the qubits in the computational basis.
This constant-depth circuit produces one bit of mutual information for each pair of qubits.
This proves that, using classical mutual information as a measure of correlation,
we would not be able to derive a sufficiently good upper bound on the amount of correlations as a function of the circuit depth.
The entanglement entropy, which is a standard measure of correlation in a bipartite quantum system suffers from the same issue.

To avoid this issue of classical communication introducing mutual information we first transform the circuit $C$ into a circuit $C'$ (where $C'$ is guaranteed to have the same action and a similar depth to $C$) by postponing the measurements and the classical controlled Pauli operations until the end of the circuit.
Moreover, to avoid the increase of mutual information by the controlled-Pauli operations at the end of the first run of $C'$, we use the conditional mutual information 
\begin{align*}
I(O_L^{(2)}, E_L; O_R^{(2)}, E_R | O^{(1)}),
\end{align*}
conditioned on the outcomes of the first run of $C'$.

This notion of correlation is adapted to obtain both a non-trivial lower bound and a non-trivial upper bound on the amount of correlation between two subsets of qubits $L$ and $R$ and leads to the bound of Theorem~\ref{theorem:separator_bound}.

The proof of Theorem~\ref{theorem:separator_bound} relies on Lemma~\ref{lemma:entropic_lower_bound} and Lemma~\ref{lemma:entropic_upper_bound} proven in the following subsections which provide a lower bound and an upper bound on the conditional mutual information $I(O^{(2)}_L, E_L; O^{(2)}_R, E_R | O^{(1)})$.
Instead of using the circuit $C \circ E \circ C$, we use the version of this circuit obtained through the circuit transformations described in Lemma~\ref{lemma:circuit_transformation} and Lemma~\ref{lemma:circuit_transformation_double_meas}.

\subsection{Circuit transformation}
\label{subsec:circuit_transformation}

Instead of working with the Clifford circuit $C$, we build a modified version $C'$ of $C$ which implements the same Pauli measurements on the data qubits but which is easier to analyze.

\begin{figure}[t]
  \centering
  \includegraphics[scale=.7]{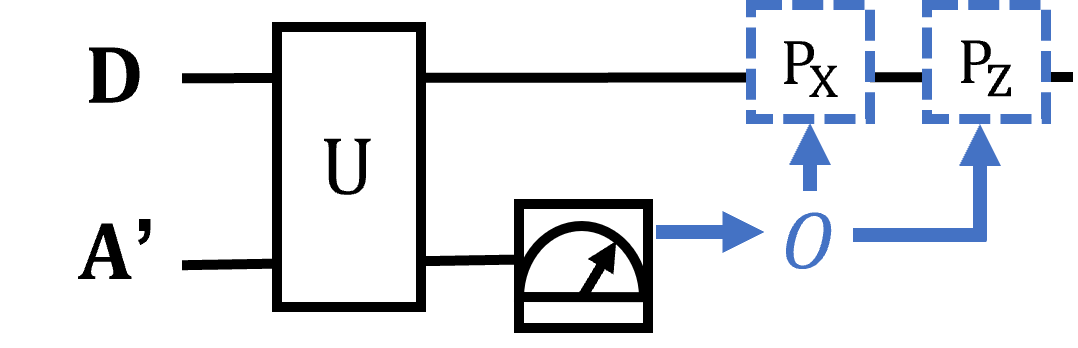}
  \caption{
  Transformed Clifford circuit obtained in Lemma~\ref{lemma:circuit_transformation}.
  Any Clifford circuit can be transformed into such a four-stage circuit made with layers of unitary Clifford gates, followed by a layer of single-qubit measurements and two layers of classically-controlled Pauli operations.
  In this figure, each wire represents a block of multiple qubits, $U$ is a product a single-qubit and two-qubit Pauli operations, the blocks $P_X$ and $P_Z$ represent a layer of classically-controlled single-qubit Paulis.
  }
  \label{fig:proof_circuit_transformation_1}
\end{figure}

\begin{lemma}\label{lemma:circuit_transformation}
Let $C$ be a Clifford circuit implementing an operation on a set of data qubits {\bf D} using a set of ancilla qubits {\bf A}.
Then, there exists a circuit $C'$ represented in Fig.~\ref{fig:proof_circuit_transformation_1} which implements the same operation as $C$ on {\bf D},
using a set of ancilla qubits ${\bf A'}$, made with the following steps.
\begin{enumerate}
\item The preparation of all the ancilla qubits in ${\bf A'}$ is in the state $\ket 0$ or $\ket +$.
\item A unitary Clifford circuit $U$ with $\depth(U) \leq 4\depth(C)$ built from single-qubit and two-qubit unitary Clifford gates.
\item A layer of single-qubit measurements of every qubit in ${\bf A'}$ with outcome set ${O}$, followed by classically-controlled $X$ and $Z$ on ${\bf D}$ depending on $O$.
\end{enumerate}
Moreover, any subset $L \subset {\bf D} \cup {\bf A}$ of qubits of $C$ maps onto a subset $L' \subset {\bf D} \cup {\bf A'}$ such that each layer of the circuit $C'$ contains at most $|\partial L|$ gates supported on both $L'$ and its complement.
\end{lemma}

Therein, the notation $\partial L$ refers to boundary edges relative to the connectivity graph of the circuit $C$.

\begin{proof}
We start with the circuit $C$ and carry out a sequence of modifications which do not change the action of the circuit until we reach $C'$.
First we eliminate two-qubit measurements.
Each joint measurement is replaced by a pair of two-qubit gates and a single-qubit measurements using the identify in Fig.~\ref{fig:circuit_simplification}.
This transformation increases the circuit depth by at most a factor four.

Given $L$, we initially define $L' = L$.
If the joint measurement is supported on two qubits of $L$ (respectively $L^{\mathsf{c}}$), the new ancilla qubit is added to $L'$ (respectively $L'^{\mathsf{c}}$).
If the joint measurement is supported on an edge of $\partial L$, then we add the ancilla qubit to $L'$.
By definition of $L'$, each layer of the circuit after this transformation contains at most $|\partial L|$ gates which act non-trivially on both $L'$ and its complement.

\begin{figure}[t]
  \centering
  \includegraphics[scale=.6]{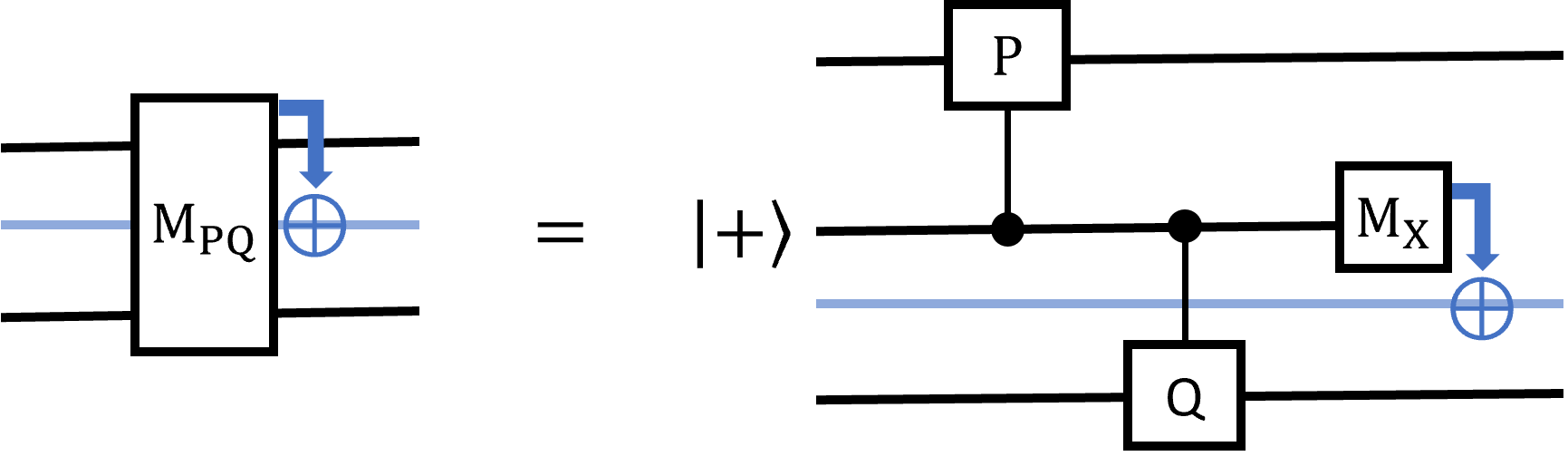}
  \caption{
  Simulation of the measurement of a two-qubit Pauli operator $PQ$ using unitary Clifford gates, single-qubit measurements and an ancilla qubit.
  }
  \label{fig:circuit_simplification}
\end{figure}

Then, each ancilla qubit $q_i$ is replaced by a set of $s_i$ ancilla qubits $q_{i, 1}, \dots, q_{i, s_i}$ where $s_i$ is the number of times $q_i$ is measured in the original circuit.
These ancilla qubits are all initialized in $\ket +$ and at most one of them can be involved in a non-trivial operation in a given time step.
Let $t_{i, 1}, \dots, t_{i, i_s}$ be the time steps where $q_i$ is measured and denote $t_{i, 0} = 0$.
The ancilla qubit $q_{i, j}$ plays the role of the qubit $q_i$ in the original circuit during all the time steps from $t_{i, j-1}+1$ to $t_{i, j}$ (included).
This transformation guarantees that each ancilla qubit is measured exactly once.
If the original ancilla qubit $q_i$ is in $L$ (respectively $L^{\mathsf{c}}$), we assign all its copies $q_{i, 1}$ to $L'$ (respectively $L'^{\mathsf{c}}$).

The number of gates supported on both $L'$ and its complement within a layer of the circuit is  unchanged during this transformation. This is because only one of the copies of an ancilla qubit is used at a given time step.

\begin{figure}[t]
  \centering
  \includegraphics[scale=.5]{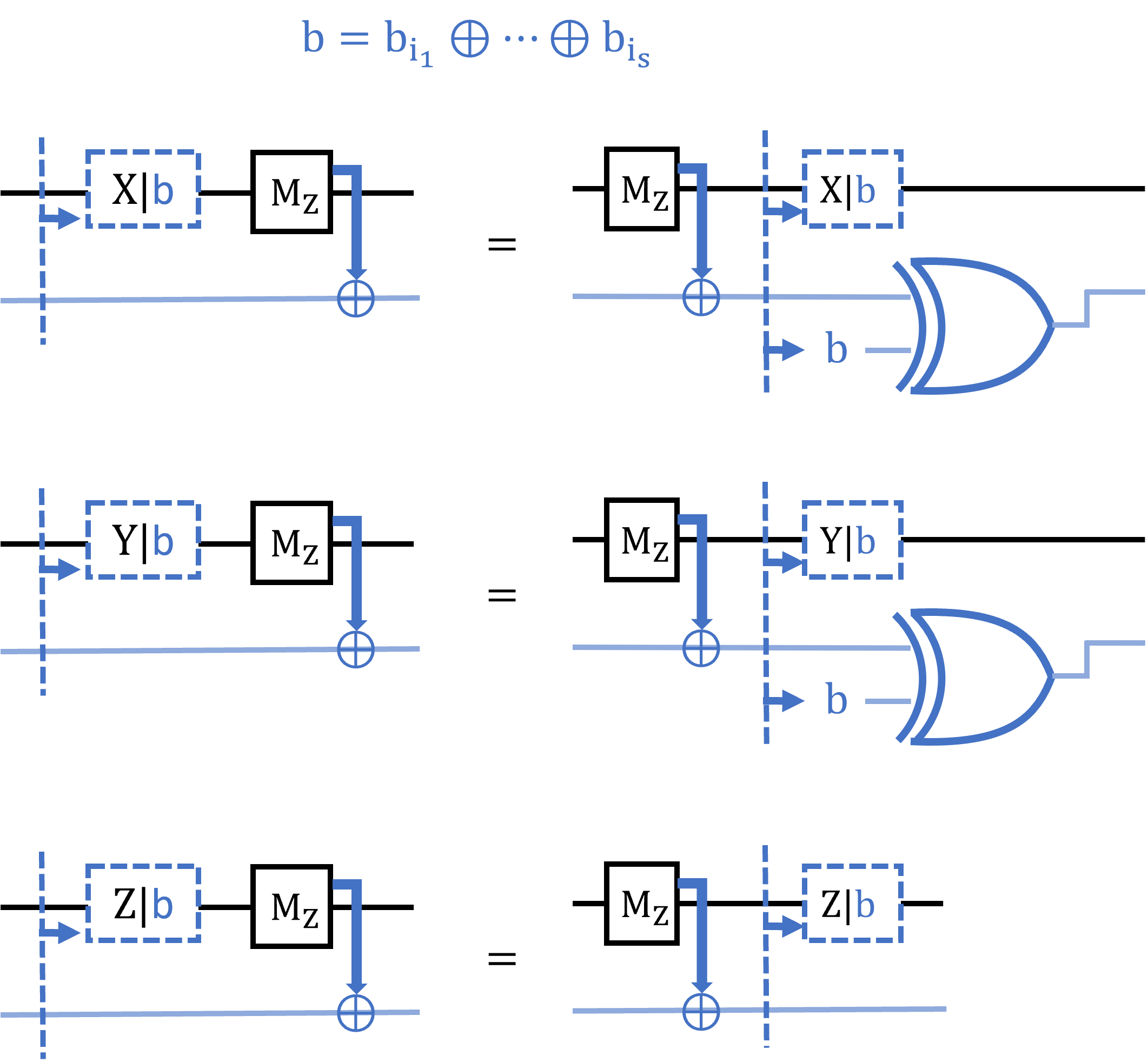}
  \caption{
  Effect of the commutation of a Pauli measurement and a classically-controlled Pauli operation.
  If the Pauli operation anti-commutes with the measured operator, swapping them introduces a XOR of the measurement outcome with the condition.
  }
  \label{fig:circuit_identities_1}
\end{figure}

\begin{figure}[t]
  \centering
  \includegraphics[scale=.5]{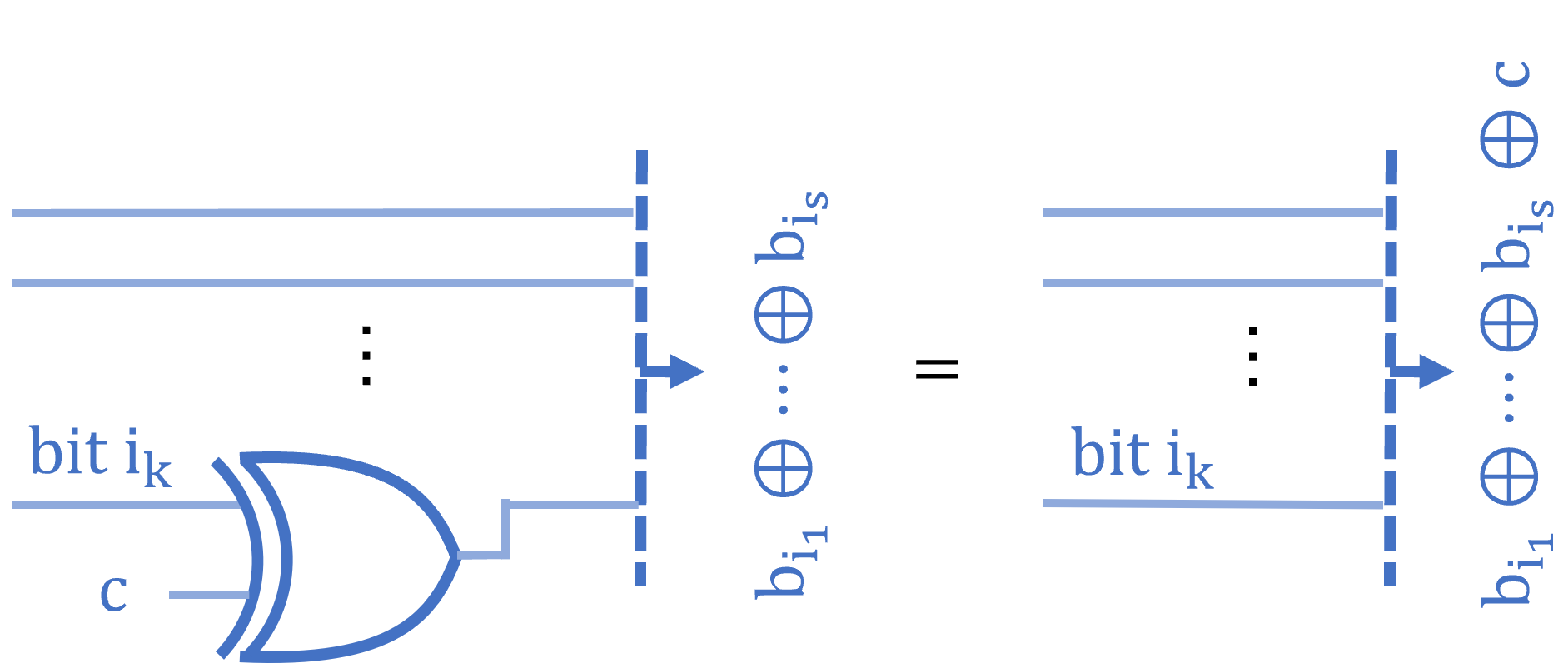}
  \caption{
  Effect of the XOR of a stored bit on a circuit output.
  }
  \label{fig:circuit_identities_2}
\end{figure}

\begin{figure}[t]
  \centering
  \includegraphics[scale=.5]{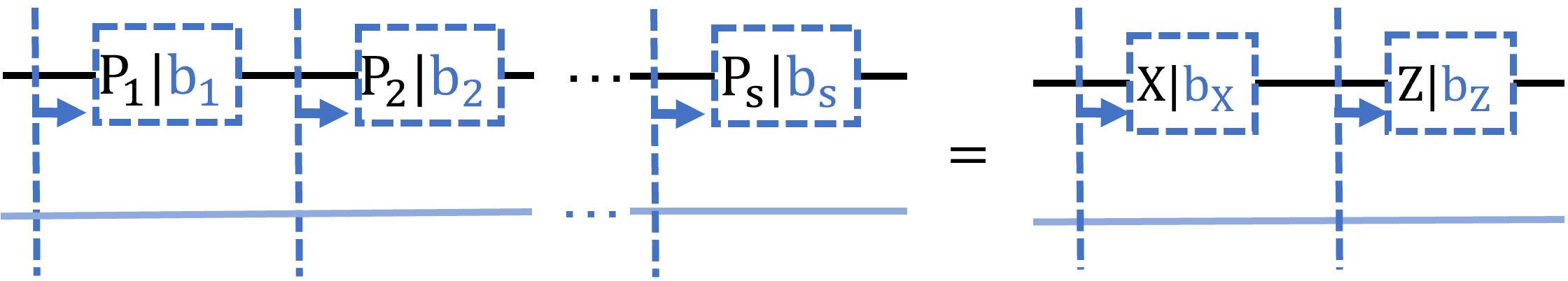}
  \caption{
  Up to a global phase, any sequence of consecutive classically-controlled Pauli operators acting on a single qubit can be decomposed into a classically-controlled $X$ followed by a classically-controlled $Z$. The conditions $b_X$ is $b_X = \oplus_{i \in I_X} b_i$ and $b_Z = \oplus_{i \in I_Z} b_i$ where $I_X$ is the set of indices $i$ such that $P_i = X$ or $Y$ and $I_Z$ is the set of indices such that $P_i = Z$ or $Y$.
  }
  \label{fig:circuit_identities_3}
\end{figure}

Now, we move all the classically-controlled Pauli operations to the end of the circuit. 
One can move a Pauli operation $P$ passed a Clifford gate $g$ using the relation $P g$ = $gg^{-1}Pg = gQ$ where $Q = g^{-1}Pg$ is also a Pauli operation. 
To move a classically-controlled Pauli operation passed a measurement, one can use the relations of Fig.~\ref{fig:circuit_identities_1}.
Then, we combine the sequences of classically-controlled Paulis into a layer of classically-controlled $X$ and a layer of classically-controlled $Z$ using the identity of Fig.~\ref{fig:circuit_identities_3}.
The classically-controlled Pauli operations at the end of the new circuit after the ancillas have been measured can be decomposed into those acting entirely on data qubits and those acting entirely on ancilla qubits. 
We discard those acting on the ancillas because they only change the parities used to produce the circuit outputs as explained in Fig.~\ref{fig:circuit_identities_2}.
These moves and other circuit transformations keep $L'$ invariant.

Finally, we can trivially postpone the measurements to perform them after the unitary gates and before the classically-controlled Pauli operations because ancilla qubits are not reused after measurement in this transformed circuit.
We can also implement all the measurements simultaneously for the same reason.
Again, the set $L'$ is kept unchanged.
\end{proof}

\subsection{The double measurement circuit}
\label{subsec:double_meas_circuit}

In this section, we consider the circuit $C \circ E \circ C$ which runs a circuit $C$ which measures Pauli operators $S_1, \dots, S_r$ followed by by a uniformly drawn random Pauli error $E$ on the data qubits before running $C$ again.
In the following lemma, we form a simplified version of the circuit $C \circ E \circ C$ using the transformation of Lemma~\ref{lemma:circuit_transformation}.

\begin{lemma}\label{lemma:circuit_transformation_double_meas}
Let $C$ be a Pauli measurement circuit on a set of data qubits {\bf D} using a set of ancilla qubits {\bf A}.
Then, there exists a circuit $\CEC$ represented in Fig.~\ref{fig:proof_circuit_transformation_2} which implements the same operation as $C \circ E \circ C$ on {\bf D}, using a set of ancilla qubits ${\bf \bar A} = {\bf A_E} \cup {\bf A_1} \cup {\bf A_2}$, made with the following steps.
\begin{enumerate}
\item The preparation of all the ancilla qubits in ${\bf A_E}$ in the state $\ket +$ and all the other ancilla in either the state $\ket 0$ or $\ket +$.
\item A unitary operation $U$ on ${\bf D \cup A_1}$ with $\depth(U) \leq 4\depth(C)$ made with single-qubit and two-qubit unitary Clifford gates.
\item A uniform random Pauli error generated by measuring the ancilla qubits of ${\bf A_E}$ in the computational basis and applying conditional Pauli error based on these measurement outcomes.
\item The same unitary operation $U$ applied to ${\bf D \cup A_2}$.
\item A layer of single-qubit measurements of all the qubits in ${\bf A_1}$ with outcome set ${O^{(1)}}$, followed by classically-controlled $X$ and $Z$ on ${\bf D \cup A_2}$ depending on $O^{(1)}$.
\item A layer of single-qubit measurements of all the qubits in ${\bf A_2}$ with outcome set ${O^{(2)}}$, followed by classically-controlled $X$ and $Z$ on ${\bf D}$ depending on $O^{(2)}$.
\end{enumerate}
Moreover, any subset $L \subset {\bf D} \cup {\bf A}$ of qubits of $C$ maps onto a subset $\bar L \subset {\bf D} \cup {\bf \bar A}$ such that each layer of the circuit $\CEC$ contains at most $|\partial L|$ gates supported on both $\bar L$ and its complement.
\end{lemma}

Like in Lemma~\ref{lemma:circuit_transformation}, the notation $\partial L$ refers to boundary edges relative to the connectivity graph of the circuit $C$.

\begin{figure}[t]
  \centering
  \includegraphics[scale=.45]{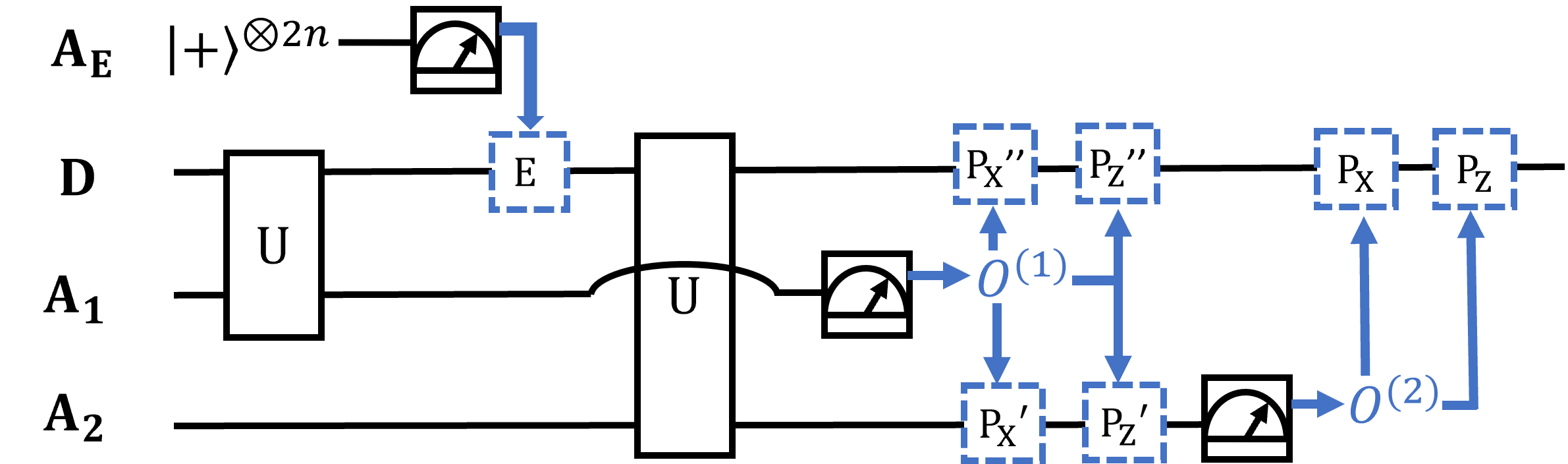}
  \caption{
  The circuit $\CEC$ obtained from $C \circ E \circ C$ by applying Lemma~\ref{lemma:circuit_transformation} to each of the two rounds of $C$.
    In this figure, each wire represents a block of multiple qubits, $U$ is a product a single-qubit and two-qubit Pauli operations, the blue dashed blocks represent layers of classically-controlled single-qubit Paulis.
  }
  \label{fig:proof_circuit_transformation_2}
\end{figure}

\begin{proof}
Starting from the circuit $C \circ E \circ C$, we build the circuit $C' \circ E \circ C'$ by applying the transformation of Lemma~\ref{lemma:circuit_transformation} to $C$.
Each round uses $|{\bf A'}|$ ancilla qubits. We assume that the ancilla qubits are not reused and we denote by ${\bf A_1}$ and ${\bf A_2}$ the set of ancilla qubits used by the first and the second application of $C'$.
The circuit includes a random Pauli error $E$ acting on the $n$ data qubits which we generate with classically-controlled Pauli operations which depend on the outcomes of the measurement of $2n$ $\ket +$ states in the computational basis.
Each data qubit $q$ corresponds to a pair of $\ket +$ states $a_X, a_Z$. 
We apply an $X$ gate (respectively a $Z$ gate) to the qubit $q$ controlled on the value of the outcome of the measurement of $a_X$ (respectively $a_Z$).

Finally, we move the first two rounds of classically-controlled Pauli operations associated with the application of the first measurement circuit to after the second block of unitary gates $U$. This is done by conjugating $P_X$ and $P_Z$ by $U^{-1}$ and separating the $X$ part and the $Z$ part as using the relation of Fig.~\ref{fig:circuit_identities_3}, which results in the Pauli operations $P'_X, P'_Z, P''_X$ and $P''_Z$ in Fig.~\ref{fig:proof_circuit_transformation_2}.

Given $L \subset {\bf D} \cup {\bf A}$, the application of Lemma~\ref{lemma:circuit_transformation} to each copy of $C$ maps $L$ onto a subset $L'$ of ${\bf D} \cup {\bf A_1} \cup {\bf A_2}$ and it guarantees that each layer contains at most $|\partial L|$ gates acting non-trivially on both $L'$ and its complement.
The set $\bar L$ is obtained from $L'$ by adding the ancilla qubits of ${\bf A_E}$ used to generate errors on the data qubits of $L$.
Moving the classically-controlled gates cannot introduce any operation acting on $\bar L$ and its complement because these gates are decomposed into single-qubit Pauli operations.
\end{proof}

\subsection{Notation}
\label{subsec:proof_notations}

In what follows, we consider a Pauli measurement circuit $C$ for the measurement of a set of Pauli operators $S_1, \dots, S_r$.
The circuit $\CEC$ is the circuit obtained in Lemma~\ref{lemma:circuit_transformation_double_meas} by simplifying the circuit $C \circ E \circ C$ where $E$ is a round of Pauli errors.
We refer to the circuit $\CEC$, represented in Fig.~\ref{fig:proof_circuit_transformation_2}, as the double measurement circuit.

Denote by $O^{(1)}$ or $O^{(2)}$ the sets of outcomes extracted during the two runs of $C$.
The circuit outputs the values $m_i^{(t)}$ where $i=1, \dots, r$ and $t=1, 2$.
The output $m_i^{(t)}$ is the outcome of the measurement of the Pauli operators~$S_i$ produced by the $t$ th run of $C$.
It is obtained as
$$
m_i^{(t)} = \bigoplus_{o \in O_i^{(t)}} o,
$$
for some subset $O_i^{(t)}$ of $O^{(t)}$.

In what follows, we consider a partition of the qubit set ${\bf D} \cup {\bf A}$ as $L \cup R$ with $R = L^\mathsf{c}$.
The subset of the qubit set ${\bf D} \cup {\bf \bar A}$ of the double measurement circuit $\CEC$ induced by $L$ is denoted $\bar L$ and its complement is denoted $\bar R$. Throughout the proof, when we use the notation $\bar L$, it is always for a set $\bar L$ induced by a subset $L$ of ${\bf D} \cup {\bf A}$.

The sets of outcomes $O^{(t)}$ obtained during the first and the second application of $C$ split along the partition $\bar L \bar R$ as $O^{(t)} = O^{(t)}_{\bar L} \cup O^{(t)}_{\bar R}$.
We use the notation $E_{\bar L}$ and $E_{\bar R}$ for the restrictions of the Pauli error $E$ to $\bar L$ and $\bar R$ respectively.

\subsection{Lower bound on the mutual information}

Here, we derive a lower bound on the mutual information between the outcomes of the double measurement circuit $\CEC$ obtained on each side of a partition of the qubits into two subsets. We use the notations of Fig.~\ref{fig:proof_circuit_transformation_2} and Section~\ref{subsec:proof_notations}.

\begin{lemma} \label{lemma:entropic_lower_bound}
With the notations of Section~\ref{subsec:proof_notations}, we have
$$
I(O_{\bar L}^{(2)}, E_{\bar L}; O_{\bar R}^{(2)}, E_{\bar R} | O^{(1)}) \geq n_{\cut}/2,
$$
where $n_{cut}$ is the number of operators $S_i$ supported on both $L$ and $R$.
\end{lemma}

\begin{proof}
By the data processing inequality, we have 
\begin{align}
& I(O_{\bar L}^{(2)}, E_{\bar L}; O_{\bar R}^{(2)}, E_{\bar R} | O^{(1)}) \\
& \geq I(M_{\bar L}^{(2)}, E_{\bar L}; M_{\bar R}^{(2)}, E_{\bar R} | O^{(1)}),
\end{align}
where $M_{\bar L}$ and $M_{\bar R}$ are the parities corresponding to the restrictions of the sets $O_i$ to each subset $\bar R$ and $\bar L$. 
For instance, $M_{\bar L}$ is the set of values $m_{i, {\bar L}}^{(t)} = \oplus_{o \in O_i^{(t)} \cap O_{\bar L}} o$.
Then, using the relation $I(A; B | C) = H(A | C) - H(A | B, C)$, we obtain
\begin{align}
& I(M_{\bar L}^{(2)}, E_{\bar L}; M_{\bar R}^{(2)}, E_{\bar R} | O^{(1)}) \label{eq:lemma_lower_bound:cond_entropy_term0} \\
& = H(M_{\bar L}^{(2)}, E_{\bar L} | O^{(1)}) 
\label{eq:lemma_lower_bound:cond_entropy_term1} \\
& \phantom{=} - H(M_{\bar L}^{(2)}, E_{\bar L} | M_{\bar R}^{(2)}, E_{\bar R}, O^{(1)}) \cdot \label{eq:lemma_lower_bound:cond_entropy_term2}
\end{align}

Consider the term $H(M_{\bar L}^{(2)}, E_{\bar L} | M_{\bar R}^{(2)}, E_{\bar R}, O^{(1)})$ in Eq.~\eqref{eq:lemma_lower_bound:cond_entropy_term2}.
Let us show that the values in $M_{\bar L}^{(2)}$ are fully determined by $O^{(1)}, M_{\bar R}^{(2)}$ and $E$. 
We have
$
m_{i}^{(2)} = m_{i}^{(1)} + m_i(E) + m_i(P')
$ 
where $m_i(E)$ is $0$ if $E$ and $S_i$ commute and $1$ otherwise and $m_i(P')$ is defined similarly for the product $P'$ of the conditional operations applied before the measurement of ${\bf A_2}$.
The conditional operation $P'$, and therefore $m_i(P')$, depends only on $O^{(1)}$.
Splitting the other term along the partition, this implies 
\begin{align} \label{eq:proof_lemma_cmi_lower_bound:outcome_decomposition}
& m_{i, {\bar L}}^{(2)} + m_{i, {\bar R}}^{(2)} \\
& = m_{i, {\bar L}}^{(1)} + m_{i, {\bar R}}^{(1)} + m_i(E_{\bar L}) + m_i(E_{\bar R})  + m_i(P'),
\end{align}
which proves that $m_{i, {\bar L}}^{(2)}$ can be obtained from $O^{(1)}, M_{\bar R}^{(2)}$ and $E$.
As a result, we get
\begin{align}
& H(M_{\bar L}^{(2)}, E_{\bar L} | M_{\bar R}^{(2)}, E_{\bar R}, O^{(1)}) \\
& = H(E_{\bar L} | M_{\bar R}^{(2)}, E_{\bar R}, O^{(1)}) \\
& = H(E_{\bar L}) \cdot
\end{align}
because $E_{\bar L}$ is independent of $M_{\bar R}^{(2)}, E_{\bar R}$ and $O^{(1)}$.

Injecting this in Eq.~\eqref{eq:lemma_lower_bound:cond_entropy_term0}, we find
\begin{align}
& I(M_{\bar L}^{(2)}, E_{\bar L}; M_{\bar R}^{(2)}, E_{\bar R} | O^{(1)}) \\
& = H(M_{\bar L}^{(2)}, E_{\bar L} | O^{(1)}) - H(E_{\bar L}) \\
& = H(M_{\bar L}^{(2)}, E_{\bar L}, O^{(1)}) - H(O^{(1)}) - H(E_{\bar L}).
\end{align}
Given that $O^{(1)}$ and $E_{\bar L}$ are independent, we have 
$
H(O^{(1)}) + H(E_{\bar L}) = H(O^{(1)}, E_{\bar L})
$
which leads to
\begin{align}
& I(M_{\bar L}^{(2)}, E_{\bar L}; M_{\bar R}^{(2)}, E_{\bar R} | O^{(1)}) \\
& = H(M_{\bar L}^{(2)}, E_{\bar L} | O^{(1)}, E_{\bar L}) \\
& = H(M_{\bar L}^{(2)} | O^{(1)}, E_{\bar L}) \cdot
\end{align}

To obtain a lower bound on this quantity, we introduce the set $S_{\cut}$ which is a maximum set of independent operators $S_i$ that have support on both $L$ and $R$.
By definition, we have $|S_{\cut}| = n_{\cut}$.
Let $S_{\cut, {\bar L}}$ (resp. $S_{\cut, {\bar R}}$) be the subset of $S_{\cut}$ that contains the operators that depend on at least one outcome in $O_{\bar L}$ (resp. $O_{\bar R}$). 
Clearly, $S_{\cut} = S_{\cut, {\bar L}} \cup S_{\cut, {\bar R}}$.

Denote by $M_{{\bar L}, \cut}^{(2)}$ the subset of $M_{\bar L}^{(2)}$ corresponding to the outcomes of operators in $S_{\cut, {\bar L}}$.
By the data processing inequality, we have
\begin{align}
H(M_{\bar L}^{(2)} | O^{(1)}, E_{\bar L})
& \geq H(M_{{\bar L}, \cut}^{(2)} | O^{(1)}, E_{\bar L}) \cdot
\end{align}
Conditioning can not increase the entropy therefore
\begin{align}
& H(M_{{\bar L}, \cut}^{(2)} | O^{(1)}, E_{\bar L}) \\
& \geq H(M_{{\bar L}, \cut}^{(2)} | O^{(1)}, E_{\bar L}, M_{\bar R}^{(2)}) \cdot
\end{align}
Finally, based on Eq.~\eqref{eq:proof_lemma_cmi_lower_bound:outcome_decomposition}, the data $m_{i, {\bar L}}^{(2)}$ given $O^{(1)}, M_{\bar R}^{(2)}$ and $E_{\bar L}$
is equivalent to the data $m_i(E_{\bar R})$.
Here again we use the fact that $P'$ is fully determined by $O^{(1)}$.
This produces
\begin{align}
& H(M_{{\bar L}, \cut}^{(2)} | O^{(1)}, E_{\bar L}, M_{\bar R}^{(2)}) \\
& = H(m_i(E_{\bar R}), S_i \in S_{\cut, {\bar L}} | O^{(1)}, E_{\bar L}, M_{\bar R}^{(2)}) \\
& = H(m_i(E_{\bar R}), S_i \in S_{\cut, {\bar L}})\cdot
\end{align}
The last equality is due to the fact that $m_i(E_{\bar R})$ is independent of $O^{(1)}, M_{\bar R}^{(2)}$ and $E_{\bar L}$.

To compute this entropy, we introduce the function $\sigma_{\cut, L}$ that maps a right-side error $E_{\bar R}$ onto the binary vector with values $m_i(E_{\bar R})$ for $S_i \in S_{\cut, {\bar L}}$.
It is a linear map.
Moreover, by definition of $S_{\cut, {\bar L}}$, it is surjective because the operators of $S_{\cut}$ are assumed to be independent and there exists at least one error $E_{\bar R}$ which anti-commutes with any $S_i \in S_{\cut, {\bar L}}$.
Therefore, by linearity the preimage of any vector has the same size as the kernel of $\sigma_{\cut, L}$.
Given that the distribution of $E_{\bar R}$ is uniform, this proves that the vector with coefficient $m_i(E_{\bar R})$ for $S_i \in S_{\cut, {\bar L}}$ is uniformly distributed, which yields
\begin{align}
H(m_i(E_{\bar R}), S_i \in S_{\cut, {\bar L}}) = |S_{\cut, {\bar L}}| \cdot
\end{align}
To conclude the proof, note that the same result holds by swapping the role of $L$ and $R$, which leads to
\begin{align}
& I(O_{\bar L}^{(2)}, E_{\bar L}; O_{\bar R}^{(2)}, E_{\bar R} | O^{(1)}) \\
& \geq \max(|S_{\cut, {\bar L}}|, |S_{\cut, {\bar R}}|).
\end{align}
This lower bound is at least $|S_{\cut}|/2$ because $S_{\cut} = S_{\cut, {\bar L}} \cup S_{\cut, {\bar R}}$.
\end{proof}

\subsection{Upper bound on the mutual information}

We establish the following bound using a strategy similar to~\cite{kull2019spacetime}.

\begin{lemma} \label{lemma:entropic_upper_bound}
With the notations of Section~\ref{subsec:proof_notations}, we have
$$
I(O_{\bar L}^{(2)}, E_{\bar L}; O_{\bar R}^{(2)}, E_{\bar R} | O^{(1)}) \leq 32 |\partial L|\depth(C) \cdot
$$
\end{lemma}

This lemma is proven by studying the growth of the mutual information in the quantum double measurement circuit $\CEC$ introduced in Section~\ref{subsec:double_meas_circuit} and represented in Fig.~\ref{fig:proof_circuit_transformation_2}.
We use the following bound which is a straightforward application of Proposition~2 from \cite{marien2016entanglement}.

\begin{proposition} \label{prop:qmi_rate}
Consider a partition of the set of qubits into two subsets $L$ and $R$.
Let $\rho$ be a $n$-qubit density matrix and let $\tilde \rho$ be the density matrix of the system after an operation $g$.
If $g$ is a two-qubit unitary gate acting on a qubit of $L$ and a qubit of $R$, then the quantum mutual information between $L$ and $R$ satisfies
\begin{align*}
S(\tilde \rho_L; \tilde \rho_R)
\leq S(\rho_L; \rho_R) + 4 \cdot
\end{align*}
\end{proposition}

\begin{proof}
Assume that $g$ acts on qubit $i$ in $L$ and qubit $j$ in $R$. 
Applying Proposition 2 of~\cite{marien2016entanglement} with
$A=\{i\}$, $a = L \backslash A$, $B=\{j\}$, and $b = R \backslash B$, 
we find
\begin{align} \label{eq:lemma_entanglement_growth_L}
S(\tilde \rho_{L}) \leq S(\rho_{L}) + 2,
\end{align}
and 
\begin{align} \label{eq:lemma_entanglement_growth_R}
S(\tilde \rho_{R}) \leq S(\rho_{R}) + 2 \cdot
\end{align}
By definition, we have 
\begin{align}
S(\tilde \rho_L; \tilde \rho_R)
& = S(\tilde \rho_{L}) + S(\tilde \rho_{R}) - S(\tilde \rho_{LR}) \\
& = S(\tilde \rho_{L}) + S(\tilde \rho_{R}) - S(\rho_{LR}) \\
& \leq S(\rho_{L}) + S(\rho_{R}) - S(\rho_{LR}) + 4,
\end{align}
where the last equation is obtained using Eq.~\eqref{eq:lemma_entanglement_growth_L} and~\eqref{eq:lemma_entanglement_growth_R}.
\end{proof}

\begin{proof} [Proof of Lemma \ref{lemma:entropic_upper_bound}]
For classical conditional mutual information, conditioning can only decrease the mutual information. This leads to
\begin{align} \label{eq:lemma_mi_upper_bound:simplify_mi}
& I(O_{\bar L}^{(2)}, E_{\bar L}; O_{\bar R}^{(2)}, E_{\bar R} | O^{(1)}) \\
& \leq I(O_{\bar L}^{(2)}, E_{\bar L}; O_{\bar R}^{(2)}, E_{\bar R}).
\end{align}
In the rest of the proof, we derive an upper bound on $I(O_{\bar L}^{(2)}, E_{\bar L}; O_{\bar R}^{(2)}, E_{\bar R})$ by bounding the growth of quantum mutual information through the circuit.

Let $\rho(t_0)$ be the initial state of the double measurement circuit $\CEC$.
Initially, the quantum mutual information between $\bar L$ and $\bar R$ is trivial, that is
\begin{align}
S(\rho_{\bar L}(t_0); \rho_{\bar R}(t_0)) = 0,
\end{align}
because the input state of the circuit is a product state.

Let $\rho(t_1)$ be the state obtained after the application of $U$ to ${\bf D}$ and ${\bf A_1}$.
The operation $U$ contains two types of gates. The gates that are supported inside $\bar L$ or inside $\bar R$ leave the quantum mutual information unchanged and the gates acting on both sides of the partition increase the quantum mutual information by at most $4$ from Proposition~\ref{prop:qmi_rate}.
As a result, we have
\begin{align}
S(\rho_{\bar L}(t_1); \rho_{\bar R}(t_1)) \leq 4 \depth(U) |\partial L|,
\end{align}
because each layer of the circuit contains at most $|\partial L|$ gates supported on both sides of the partition by Lemma~\ref{lemma:circuit_transformation_double_meas}.

The application of the random Pauli error $E$ cannot increase the quantum mutual information because it can be decomposed as a pair of independent CPTP maps ${\cal E}_L, {\cal E}_R$ acting on each side of the partition. This leads to
\begin{align}
S(\rho_{\bar L}(t_2); \rho_{\bar R}(t_2)) \leq 4 \depth(U) |\partial L|,
\end{align}
where $\rho(t_2)$ is the state of the system after the application of $E$.

The second application of $U$ produces a state $\rho(t_3)$ with
\begin{align}
S(\rho_{\bar L}(t_3); \rho_{\bar R}(t_3)) \leq 8 \depth(U) |\partial L|,
\end{align}
because it increases the quantum mutual information by at most $4$ for each gate crossing the $\bar L \bar R$ partition like the previous application of $U$.

Then, each of the qubits of ${\bf A_1}$ is measured, producing the outcomes $O^{(1)}$. These single-qubit measurements, which are CPTP maps acting either on $\bar L$ or $\bar R$, cannot increase the quantum mutual information.
Therefore, we are left with a state $\rho(t_4)$ after the measurement of ${\bf A_1}$ (without discarding ${\bf A_1}$) such that 
\begin{align}
S(\rho_{\bar L}(t_4); \rho_{\bar R}(t_4)) 
& \leq 8 \depth(U) |\partial L|.
\end{align}

Before looking into the effect of the classically-controlled Pauli operations, we consider the state $\rho(t_5)$ obtained by tracing out the subsystems ${\bf A_1}$ and ${\bf D}$.
This state, supported on ${\bf A_2}$ and ${\bf A_E}$, satisfies
\begin{align}
S(\rho_{\bar L}(t_5); \rho_{\bar R}(t_5)) 
& \leq S(\rho_{\bar L}(t_4); \rho_{\bar R}(t_4)) \\
& \leq 8 \depth(U) |\partial L|,
\end{align}
because discarding a subsystem does not increase the quantum mutual information.

Consider now the classically-controlled Pauli operations controlled on $O^{(1)}$ acting on $\rho(t_5)$.
These operations can be decomposed as a product of single qubit Pauli operation controlled on the discarded outcomes of the measurement of ${\bf A_1}$.
Such an operation is a CPTP map acting either on $\bar L$ or on $\bar R$ and therefore it cannot increase the quantum mutual information which results in 
\begin{align}
S(\rho_{\bar L}(t_6); \rho_{\bar R}(t_6)) 
& \leq 8 \depth(U) |\partial L|,
\end{align}
where $\rho(t_6)$ is the state of the subsystem ${\bf A_2 A_E}$ after the Pauli operations controlled on $O^{(1)}$.

Finally, the measurement of the qubits of ${\bf A_2}$, which can again be decomposed into CPTP maps acting either on $\bar L$ or on $\bar R$, produces a quantum state $\rho(t_7)$ over ${\bf A_2 A_E}$ with
\begin{align} \label{eq:lemma_mi_upper_bound:qmi_bound}
S(\rho_{L}(t_7); \rho_{R}(t_7)) 
& \leq 8 \depth(U) |\partial L| \cdot
\end{align}

The outcome of the measurement of ${\bf A_2}$ is the set $O^{(2)}$.
Moreover, the value of the error $E$ is stored in the qubits of ${\bf A_E}$ after measurement.
Given that the final state $\rho_{\bar L}(t_7)$ over ${\bf A_2 A_E}$ was projected onto a mixture of orthogonal pure states by measurement, the quantum mutual information of $\rho_{\bar L}(t_7)$ coincides with the classical mutual information
\begin{align} \label{eq:lemma_mi_upper_bound:cmi_to_qmi}
I(O_{\bar L}^{(2)}; O_{\bar R}^{(2)}) = S(\rho_{\bar L}(t_7); \rho_{\bar R}(t_7)) \cdot
\end{align}

Putting together Eq.~\eqref{eq:lemma_mi_upper_bound:simplify_mi}, \eqref{eq:lemma_mi_upper_bound:cmi_to_qmi} and ~\eqref{eq:lemma_mi_upper_bound:qmi_bound}, we get
\begin{align}
& I(O_{\bar L}^{(2)}, E_{\bar L}; O_{\bar R}^{(2)}, E_{\bar R} | O^{(1)}) \\
& \leq I(O_{\bar L}^{(2)}, E_{\bar L}; O_{\bar R}^{(2)}, E_{\bar R}) \\
& = S(\rho_{\bar L}(t_7); \rho_{\bar R}(t_7)) \\
& \leq 8 \depth(U) |\partial L|,
\end{align}
and by Lemma~\ref{lemma:circuit_transformation_double_meas}, we have $\depth(U) \leq 4\depth(C)$, proving the result.
\end{proof}

\subsection{Proof of Theorem~\ref{theorem:separator_bound}}

Combining all the ingredients of this section, we can conclude the proof of Theorem~\ref{theorem:separator_bound}.

\begin{proof} [Proof of Theorem~\ref{theorem:separator_bound}]
By Lemma~\ref{lemma:entropic_lower_bound} and Lemma~\ref{lemma:entropic_upper_bound}, we have
\begin{align}
n_{\cut}/2 
& \leq I(O^{(2)}_{\bar L}, E_{\bar L}; O^{(2)}_{\bar R}, E_{\bar R} | O^{(1)}) \\
& \leq 32 |\partial L| \depth(C),
\end{align}
which leads to
\begin{align} 
\depth(C) \geq \frac{n_{\cut}}{64 |\partial L|},
\end{align}
proving the theorem.
\end{proof}

\section{Space-optimal circuits for bounded-depth syndrome extraction}
\label{sec::constant_depth_circuits}

\begin{figure*}[t]
\centering
\includegraphics[scale=.52]{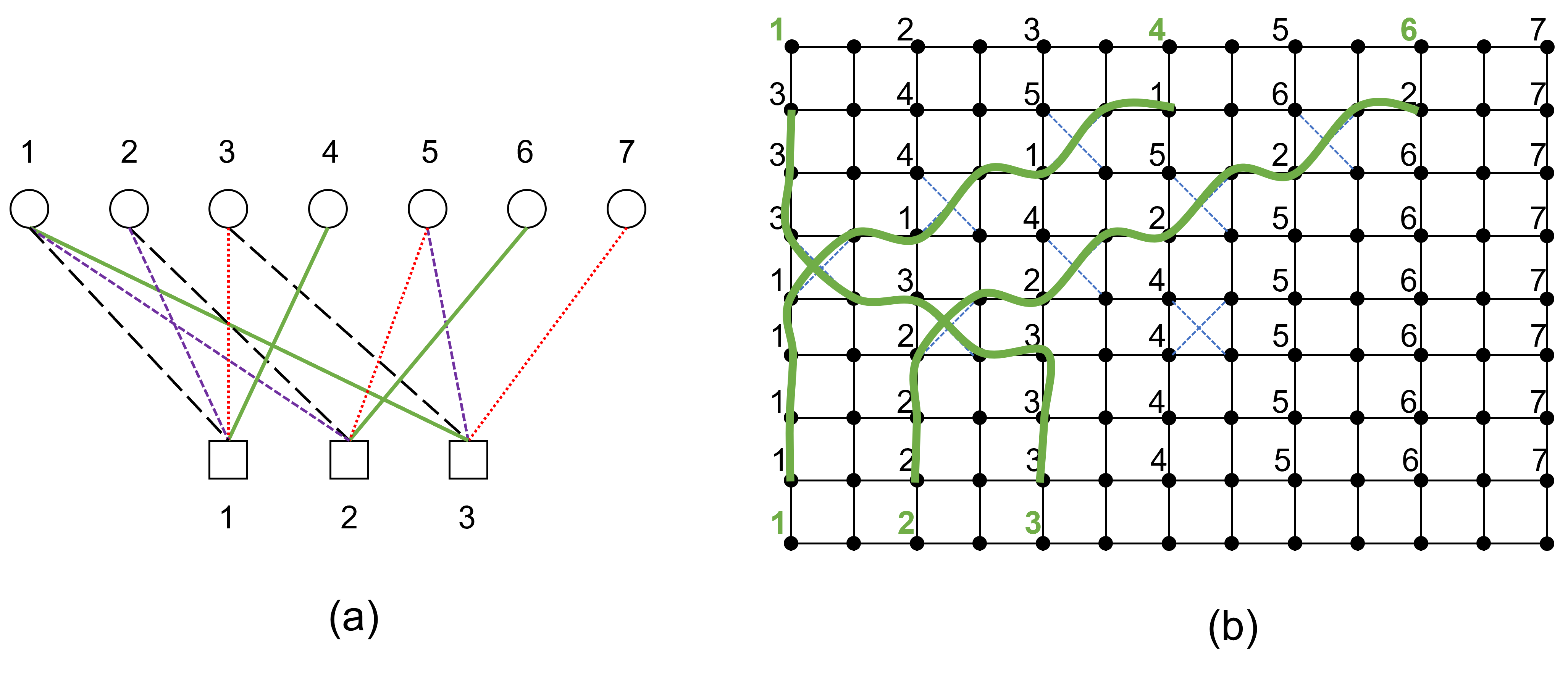}
\caption{
Constant-depth circuits to measure the $X$ stabilizers of the Steane code with full connectivity and with 2D local connectivity.
  (a)
With fully connected qubits, the syndrome can be extracted using one readout qubit per check, with four rounds of CNOT gates.
Each round corresponds to a color $c$, and the round $c$ implements the CNOT gates along all edges with color $c$ in parallel.
  (b)
In two dimensions, the seven data qubits of Steane code are mapped onto the top side of the grid and the three readout qubits are placed on the bottom.
We replace the CNOT gates of the fully connected circuit by long range CNOTs applied along paths of ancilla qubits connecting readout and data qubits.
To generate the paths with color $c$, we first place the index of each readout qubit for color $c$ below its data qubit. Then we generate the next rows of indices by applying an odd-even sorting network.
Following the index of a readout qubit, we obtain the path connecting it to its data qubit.
To simultaneously apply CNOTs along the paths with color $c$, we first prepare Bell states entangling the endpoints of crossing edges (dashed blue edges) and we use these states to apply long range CNOT gates as in Fig.~\ref{fig:long_distance_CNOT} between each readout qubit and the corresponding data qubit.
}
\label{fig:2d_constant_depth_circuit}
\end{figure*}

\begin{figure*}[ht]
\centering
\includegraphics[scale=.5]{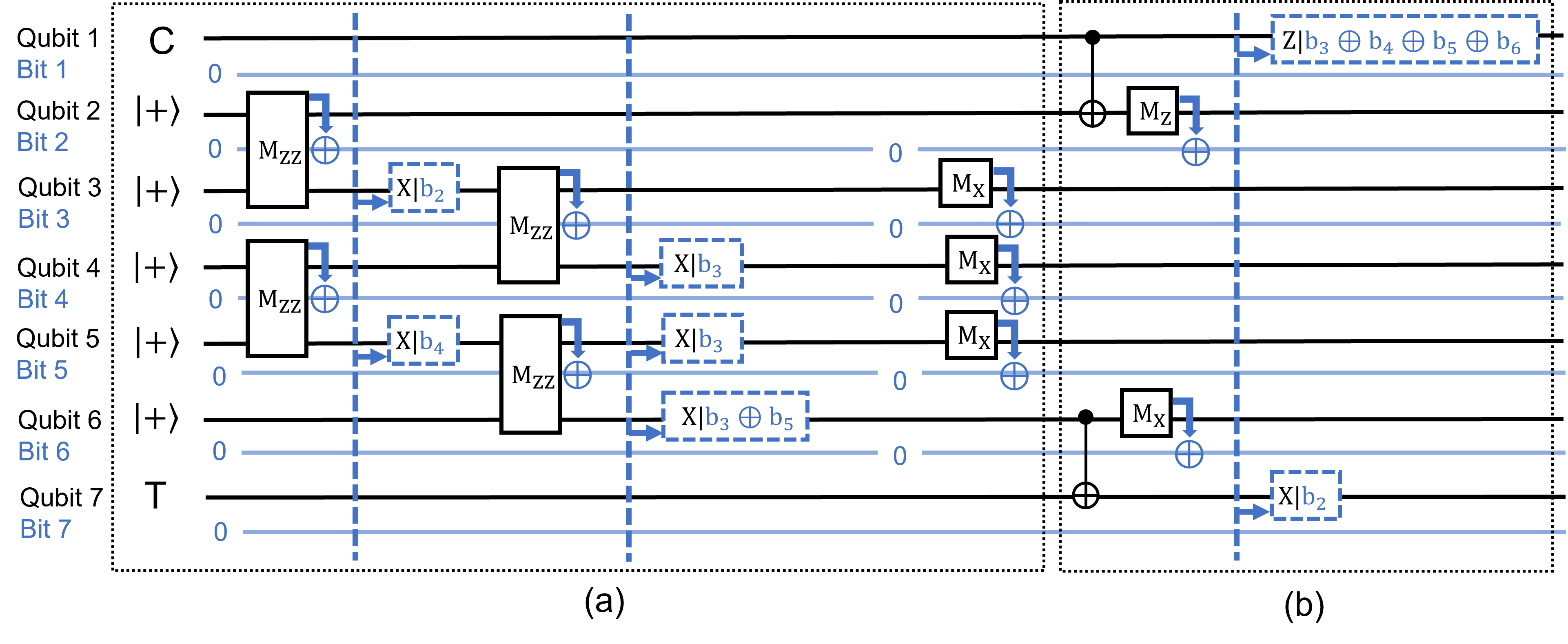}
\caption{
A long distance CNOT can be implemented in two steps. With the block (a), a Bell pair is created between two ancilla qubits, one next to the control qubit $C$ and the other one next to the target qubit $T$.
Then, this entangled state is used in the block (b) to perform a CNOT with control qubit $C$ and target $T$.
}
\label{fig:long_distance_CNOT}
\end{figure*}

Proposition~\ref{prop:dense_2d_local_meas_circuit} shows that any Clifford bounded-depth 2D local syndrome extraction circuit for a quantum expander LDPC code would require at least $\Omega(n^2)$ ancilla qubits.
In this section, we prove that this bound is tight by constructing a family of bounded-depth 2D local Clifford syndrome extraction circuits using $O(n^2)$ ancilla qubits for any CSS codes with bounded degree Tanner graph.
More precisely, we prove the following result.

\begin{proposition} \label{prop:bounded_depth_circuit_construction}
Let $Q$ be a CSS code with length $n$, with $r \leq n$ stabilizer generators.
The code $Q$ admits a 2D local Clifford syndrome extraction circuit using $a_n = 2n^2 + 2n - 2$ ancilla qubits with depth $14(\deg(T_X) + \deg(T_Z)) + 4$ where $T_X$ is the $X$ Tanner graph and $T_Z$ is the $Z$ Tanner graph of $Q$.
\end{proposition}

We first design a bounded-depth circuit for bounded degree CSS codes assuming fully connected qubits.
Then, we explain how to implement this circuit with local gates in a 2D grid of qubits.
For simplicity, we consider only the measurement of $X$ stabilizer generators. 
The same technique applies to $Z$ stabilizers.

\begin{algorithm}[h]
  \DontPrintSemicolon

  \SetKwInOut{Input}{input}\SetKwInOut{Output}{output}
  \Input{The $X$ Tanner graph $T_X$ of a CSS code.}
  \Output{A Clifford circuit for the measurement of all $X$ stabilizer generators of $T$ using fully connected qubits.}

  \BlankLine
	Compute a minimum edge coloration of $T_X$. \;
    Prepare each readout qubit in the state $\ket +$. \;
	\For{each color $c$ of $T_X$} 
 	{
    		Apply simultaneously $\cnot(s_i, q_j)$ for each edge with color $c$ connecting 
	    a readout qubit $s_i$ with a data qubit $q_j$. \;
    }
	Measure each readout qubit in the $X$ basis. \;
    	
	\caption{Fully connected syndrome extraction circuit.}
	\label{algo:fully_connected_circuit}
\end{algorithm}

In this section, we denote by $q_1, \dots, q_{n}$ the vertices of the $X$ Tanner graph $T_X$ corresponding to the $n$ data qubits and $s_1 \dots, s_{r}$ the vertices corresponding to the $r$ $X$-type stabilizer generators.
With each stabilizer vertex $s_i$, we associate an ancilla qubit that we refer to as the {\em readout qubit}.

Recall that an edge coloration is a map which associates a color with each edge in such a way that no edges sharing a vertex have the same color.
In general, computing a minimal edge-coloration for a graph, that is a coloration using a minimum number of colors, is difficult. 
However, this is easy in the case of Tanner graphs because they are bipartite, and the minimal edge-coloration is guaranteed to have at most $\deg(T_X)$ colors.
The algorithm of \cite{alon_simple_2003} returns a minimal edge-coloration 
for bipartite graphs and runs in time that is quasi-linear in the number of edges.
Assuming fully connected qubits, one can use an edge-coloration of the $X$ Tanner graph $T_X$ to measure the $X$ stabilizers with a circuit of depth $\deg(T_X)+2$ using just one ancilla per stabilizer as shown in Algorithm~\ref{algo:fully_connected_circuit}.

\begin{algorithm}[h!]
  \DontPrintSemicolon

  \SetKwInOut{Input}{input}\SetKwInOut{Output}{output}
  \Input{The $X$ Tanner graph $T_X$ of a CSS code with length $n$ and with $r \leq n$ $X$ stabilizer generators.}
  \Output{A 2D 1-local Clifford circuit for the measurement of all $X$ stabilizer generators of $T$.}

  \BlankLine
	
	Consider a grid of $(2n-1) \times (n+2)$ qubits. \;
	Place the data qubits $q_1, \dots, q_n$ on every other vertex of the top row of the grid. \;
	Place the readout qubits $s_1, \dots, s_r$ on every other vertex of the bottom row of the grid. \;
	Prepare each readout qubit in the state $\ket +$. \;
	Compute a minimum edge-coloration of $T_X$. \;
 	\For{each color $c$ of $T_X$} 
 	{
 		For $i=1, \dots, r$, denote $\{s_i, q_{j_i}\}$ the edges with color $c$. \;
 		Using a odd-even sorting network, construct a family of paths connecting $s_i$ and $q_{j_i}$ in the grid as shown in Fig.~\ref{fig:2d_constant_depth_circuit}. \;
 		Prepare a Bell state on each diagonal edge of a path. \;
 		Apply simultaneously a long-distance CNOT gate $\cnot(s_i, q_{j_i})$ along the path from $s_i$ to $q_{j_i}$, using the circuit of Fig.~\ref{fig:long_distance_CNOT}. \;
 	}
	Measure each readout qubit in the $X$ basis. \;
	\caption{Switch-based syndrome extraction circuit in a two-dimensional grid of qubits.}
	\label{algo:swicth_based_circuit}
\end{algorithm}

We now describe a mapping of this constant-depth circuit with unrestricted connectivity to a circuit with local connectivity on a two-dimensional $(2n-1) \times (n+2)$ grid as depicted in Figure~\ref{fig:2d_constant_depth_circuit}.
The $n$ data qubits $q_j$, for $j=1, 2, \dots, n$, are placed on top of the grid in position $(x(j), n+1)$ where $x(j) = 2j-2$.
The $r$ readout qubits $s_i$, with $i=1, 2, \dots, r$, are placed along the bottom side in position $(x(i), 0)$.
One could use a long-distance CNOT gate as in Figure~\ref{fig:long_distance_CNOT} 
to implement a CNOT gate between a readout qubit $s_i$ at the bottom and a data qubit $q_j$ at the top of the grid with the use of a path of ancillas that connects them.
However, this technique does not allow us to implement simultaneously all the CNOTs with color $c$ because the corresponding paths typically cross which means that they use the same ancilla qubit.
In what follows, we describe a procedure which allows for a simultaneous implementation of all these CNOT gates, leading to a bounded-depth syndrome extraction circuit.

First, for all edges of a color $c$, we generate a family of paths in the grid connecting each readout qubit $s_i$ appearing in one of the edges to the data qubit $q_{j_i}$ that the edge connects to (see Figure~\ref{fig:2d_constant_depth_circuit}). 
Consider the row below the data qubits.
We associate an index, denoted $\Index(j, n)$, with the qubits in position $(x(j), n)$ such that $\Index(j, n) \in \{1, 2, \dots, n\}$ and $\Index(j_i, n) = i$.
Then, the indices of the all the rows below are obtained by applying an odd-even sorting network~\cite{cormen2009introduction} to theses indices, as one can see in Figure~\ref{fig:2d_constant_depth_circuit}. 
To build the indices of the row $b$ with odd $b$, we copy the indices of the previous row (row $b+1$) and we swap $\Index(2a+1, b)$ and $\Index(2a+2, b)$ if $\Index(2a+1, b) > \Index(2a+2, b)$.
For even $b$, we instead swap the pairs $\Index(2a, b)$ and $\Index(2a+1, b)$ when they are in decreasing order.
When two indices are swapped, we add a pair of crossing edges (dashed gray edges) in the grid inside the corresponding face.
This sorting network produces a sorted sequence of indices after at most $n$ levels, which produces a path connecting each readout qubit  $s_i$ with its associated data qubit $q_{j_i}$.

Then, we generate Bell states along the crossing edges in a face.
A crossing pair of Bell states in the face with corners $(a, b)$ and $(a+1, b+1)$ 
can be obtained using allowed local operations in the grid with the following procedure. 
First, prepare a state $\ket +$ on each of the four vertices of the face.
Then, measure $XX$ and $ZZ$ on the two horizontal edges of the face to obtain Bell states on these edges.
Finally, swap the qubits $(a, b)$ and $(a, b+1)$ using a sequence of three CNOTs.
This circuit can be executed in parallel over all the faces supporting a crossing to produce all the required Bell states in seven steps (including a layer of classically-controlled Paulis).

A Bell pair entangling two qubits of the grid can be interpreted as a \introduce{virtual edge} connecting these two qubits.
This virtual edge can be used to perform a long distance CNOT gate~\cite{gottesman1999_gate_teleportation} with the circuit presented in Figure~\ref{fig:long_distance_CNOT}.
The only difference is that virtual edges are destroyed after being used because the long distance CNOT gates consumes the Bell states.
In the last step of the algorithm, we implement a long distance CNOT gate along each path produced by the sorting network. 
These CNOT gates can be implemented simultaneously in seven steps (including a layer of classically-controlled Paulis) because the corresponding paths do not share any qubit.

With this approach the measurement of all the $X$ stabilizer generators can be implemented in the 2D grid of qubits using a round of preparation of the readout qubits, $\deg(T_X)$ rounds of long distance CNOT gates in depth 14 and the measurement of the readout qubits.

\section{Depth-optimal circuits for linear space overhead syndrome extraction}
\label{sec::constant_overhead_circuits}

\begin{figure}[ht]
  \centering
  \includegraphics[width=0.45\textwidth]{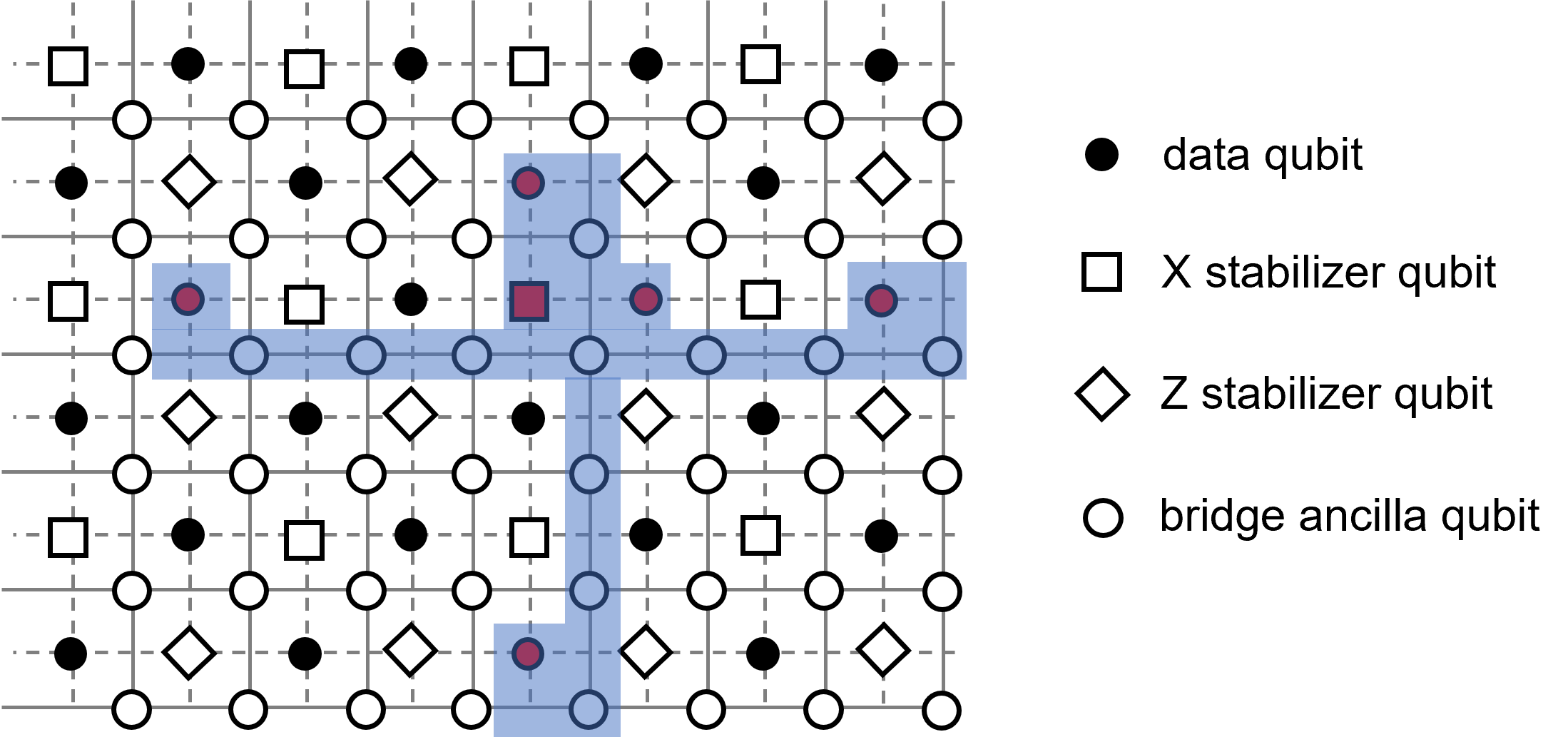}
  \caption{
	Each $X$ (or $Z$) stabilizer generators of a HGP code is supported on the union of a row and a column of a square grid of qubits and can be measured using operations supported on qubits in the vicinity of that row and column.
  	As an example, we mark in red the readout qubit and the data qubits in the support of an $X$ stabilizer generator.
	The measurement of this stabilizer is implemented using the readout qubit prepared in the state $\ket +$.
	Then, we implement a vertical multi-target CNOT controlled on the readout qubit and targeting the qubits of the support of the stabilizer on the column of the readout qubit using nearby bridge qubits.
	A similar multi-target CNOT is then implemented horizontally.
	The stabilizer outcome is extracted by measuring the readout qubit in the $X$ basis.
  }
  \label{fig:constant_overhead_layout}
\end{figure}

\begin{figure}[ht]
  \centering
  \includegraphics[scale=.5]{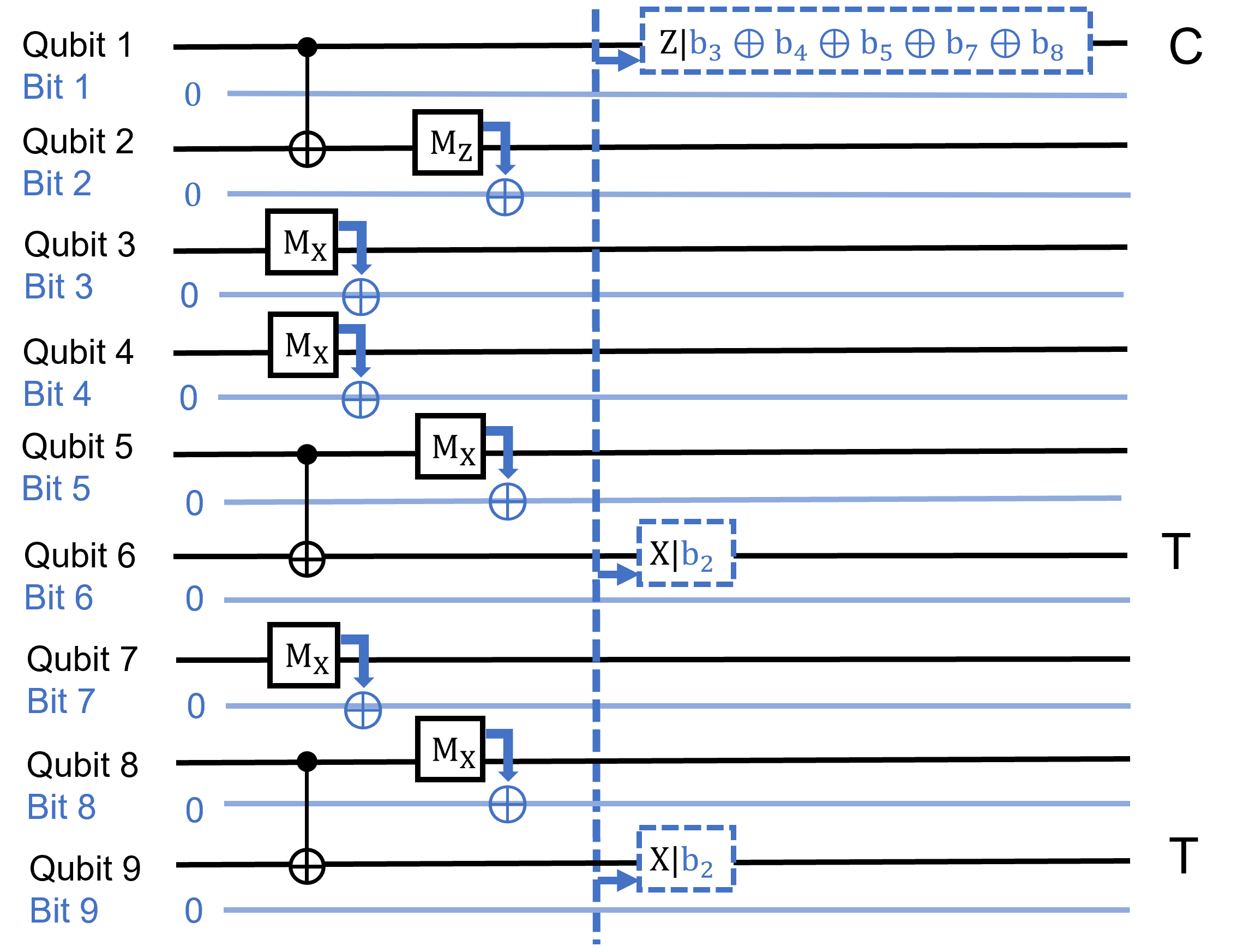}
  \caption{
  A multi-target CNOT gate implemented using a cat state with control qubit C and two target qubits T.
  We assume that qubits 2, 3, 4, 5, 7, 8 are initially in the six-qubit cat state $\frac{1}{\sqrt{2}}(\ket{000000} + \ket{111111})$.
  This cat state can be prepared using $\ket +$ states and two layers of $ZZ$ measurements and two layers of classically=-controlled Pauli operations as in Fig.~\ref{fig:long_distance_CNOT}.
  If the control qubit and the target qubits are associated with distinct cat state qubits all the CNOTs can be implemented in parallel. Then, the depth of the multi-target CNOT is eight (including the cat state preparation in depth five), independently of the number of targets.
  }
  \label{fig:multi_target_CNOT}
\end{figure}

In this section, we construct a family of 2D local Clifford syndrome extraction circuits for any hypergraph product (HGP) code~\cite{tillich_quantum_2014}.
HGP codes are one of the most popular constructions of quantum codes, where an HGP code is formed from a pair of classical codes.
Proposition~\ref{prop:dense_2d_local_meas_circuit} shows that any Clifford 2D local syndrome extraction circuit for a quantum expander LDPC code using $O(n)$ ancilla qubits has depth at least $\Omega(\sqrt{n})$.
Our 2D local syndrome extraction circuits prove this bound is tight when applied to any family of HGP codes formed from a pair of classical LDPC codes with proportional length, as they use $O(n)$ ancilla qubits and have depth $O(\sqrt{n})$ in this setting.

Recall that the Tanner graph of a classical linear code $C$ with length $n$ defined by $r$ checks is the bipartite graph $T = (V, E)$ with $V = V_B \cup V_C$ where $V_B = \{b_1, \dots, b_n\}$ corresponds to the set of bits and $V_C = \{c_1, \dotsm c_r\}$ corresponds to the checks. 
The bit $b_i$ and the check $c_j$ are connected by an edge iff the bit $b_i$ belongs to the support of the check $c_j$.
By swapping the roles of the bits and the checks, we obtain a second linear code called the {\em transposed code}, denoted $C^T$, which has length $r$ and is defined by $n$ checks.
In what follows, the parameters of a linear code are denoted $[n, k, d]$ and the parameters of its transposed code are denoted $[n^T, k^T, d^T]$.

The HGP construction, that we review below, takes as an input two classical linear codes with Tanner graphs $T_1$ and $T_2$ with parameters $[n_1, k_1, d_1]$ and $[n_2, k_2, d_2]$ and produces a quantum CSS code with parameters $[[n, k, d]]$ where $n = n_1 n_2$, $k=k_1 k_2 + k_1^T k_2^T$ and $d \geq \min\{d_1, d_2, d_1^T, d_2^T\}$.

In the rest of this section, we describe a syndrome extraction circuit that applies to any HGP codes which consumes $O(n)$ ancilla qubits, leading to the following proposition.

\begin{proposition} \label{prop:linear_space_circuit_construction}
Let $Q$ be the HGP of two classical codes $C_i$ with length $n_i$ and with $r_i$ checks for $i=1,2$.
The code $Q$ admits a 2D local Clifford syndrome extraction circuit using 
$
a_n = n_1 n_2 + 2n_1 r_2 + 2r_1n_2 + r_1r_2
$
ancilla qubits with depth 
$
8(n_1 + n_2 + r_1 + r_2) + 4.
$
\end{proposition}

Applying this result to a family of HPG codes which are the product of linear codes with $r_i = O(n_i)$ and with length $n_2 = \Theta(n_1)$, we obtain 2D local syndrome extraction circuits using $O(n_1 n_2) = O(n)$ ancilla qubits with depth $O(n_1) = O(\sqrt{n})$.
This family of quantum codes includes includes local-expander quantum codes~\cite{leverrier_quantum_2015} which obey the bound of Corollary~\ref{cor:dense_2d_local_meas_circuit} proving that our circuit construction saturates the result of this corollary for syndrome extraction circuits using $O(n)$ ancilla qubits.

Let us review the HGP construction.
Consider a pair of Tanner graphs $T_1 = (V_{1, B} \cup V_{1, C}, E_1)$ and $T_2 = (V_{2, B} \cup V_{2, C}, E_2)$.
The Tanner graph $T = (V, E)$ of the HGP of $T_1$ and $T_2$ is defined to be the Cartesian product of $T_1$ and $T_2$. 
The vertex set of $T$ is $V = V_1 \times V_2$ with an edge connecting $(u_1, u_2)$ and $(v_1, v_2)$ iff either $u_1=v_1$ and $\{u_2,v_2\} \in E_2$, or $\{u_1,v_1\} \in E_1$ and $u_2 = v_2$. 
We associate a {\em data qubit} with each vertex $(v_1, v_2) \in V_{1, C} \times V_{2, C} \cup V_{1, B} \times V_{2, B}$.
The $X$ (respectively $Z$) stabilizer generators correspond to the vertices of $V_{1, B} \times V_{2, C}$ (respectively $V_{1, C} \times V_{2, B}$).

For $i=1, 2$, denote by $v_{i, 1}, \dots, v_{i, |V_i|}$ the vertices of $V_i$.
To implement the syndrome extraction circuit in a 2D grid of qubits, we use one {\em readout qubit} for each stabilizer generator. 
The data qubit or the readout qubits associated with the vertex $(v_i, v_j)$ is placed in position $(i, j)$ of the square grid.
We use an additional set of ancilla qubits, that we refer to as \emph{bridge qubits}, placed at half-integer coordinates $(i+1/2,j+1/2)$ for $i=1, \dots, |V_1|$ and $j=1, \dots, |V_2|$.
There is a total of $n_1 n_2$ data qubits, $n_1 r_2$ $X$ readout qubits, $r_1 n_2$ $Z$ readout qubits and $(n_1 + r_1)(n_2 + r_2)$ bridge qubits.

The qubits occupy occupy a constant fraction of a $2(n_1 + r_1) \times 2(n_2 + r_2)$ patch in $Z^2$ (up to a multiplication of the coordinates by two).
We will design a circuit that uses only local operations acting on qubits at bounded distance from each other.

Before stating the explicit circuit construction in Algorithm~\ref{algo:linear_space_circuit}, we explain the intuition behind the circuit by describing the measurement of a single $X$ stabilizer generator.
By construction, any stabilizer generator is supported on the union of a row and a column of qubits with the corresponding readout qubit lying at their intersection $(i, j)$.
The $X$ stabilizer generator $S_{i, j}$ with readout qubit $(i, j)$ can be measured by first preparing its readout qubit in the state $\ket{+}$, applying a multi-target CNOT from $(i, j)$ to all the qubits in the support of $S_{i, j}$ in column $i$, then applying a multi-target CNOT from $(i, j)$ to all the qubits in the support of $S_{i, j}$ in row $j$, before finally measuring the readout qubit in the $X$ basis.
A multi-target CNOT in row $j$ (respectively column $i$) can be implemented in depth eight using the row $j+1/2$ (respectively column $i+1/2$) of bridge qubits.
This mutli-target CNOT circuit is described in Fig.~\ref{fig:multi_target_CNOT}.

To build the full stabilizer measurement circuit for $X$ stabilizer generators, we parallelize the circuit measuring a single generator by simultaneously applying mutli-target CNOTs on different rows or different columns.
If $v_i \notin V_{1, B}$ then column $i$ contains no $X$ readout qubit.
If $v_i \in V_{1, B}$, then column $i$ contains $|V_{2, C}|$ $X$ readout qubits, which each control a horizontal mutli-target CNOT.
These $|V_{2, C}|$ non-overlapping mutli-target CNOTs can be implemented in parallel.
With this strategy, the $|V_{1, B}||V_{2, C}|$ multi-target CNOTs required to perform the measurement of all the $X$ stabilizer generators can be implemented in $|V_{1, B}|$ rounds of mutli-target CNOTs.
With the same approach, the $|V_{1, B}||V_{2, C}|$ vertical multi-target CNOTs can be implemented in $|V_{2, C}|$ rounds.
Adding the preparation and measurement of the readout qubits, this produces a circuit with depth $8(|V_{1, B}| + |V_{2, C}|) + 2$ which is described in Algorithm~\ref{algo:linear_space_circuit}.
The $Z$ stabilizers are measured in depth $8(|V_{1, C}| + |V_{2, B}|) + 2$ using the same strategy.

\begin{algorithm}[h!]
  \DontPrintSemicolon

  \SetKwInOut{Input}{input}
  \SetKwInOut{Output}{output}
  \Input{A pair of Tanner graphs $T_1 = (V_{1, B} \cup V_{1, C}, E_1)$ and $T_2 = (V_{2, B} \cup V_{2, C}, E_2)$.}
  \Output{A circuit for the measurement of all $X$ stabilizer generators in a two-dimensional grid of qubits for the HGP of $T_1$ and $T_2$.}

  \BlankLine

	Prepare each $X$ readout qubit in the state $\ket +$. \;
  \For{each $i$ $\in 1,\dots, |V_1|$} {
    Simultaneously, for each $j$ $\in 1,\dots, |V_2|$ \;
		  Apply a multi-target CNOT controlled on the readout qubit at $(i,j)$ and targeted on all qubits in column $j$ that belong to the support of the corresponding stabilizer. \;
    }
   \For{each $j$ $\in 1,\dots, |V_2|$} {
    Simultaneously, for each $i$ $\in 1,\dots, |V_1|$ \;
		  Apply a multi-target CNOT controlled on the readout qubit at $(i,j)$ and targeted on all qubits in row $i$ that belong to the support of the corresponding stabilizer. \;
    }
    	Measure each $X$ readout qubit in the $X$ basis. \;
  
	\caption{A linear space-overhead syndrome extraction circuit for HGP codes.}
	\label{algo:linear_space_circuit}
\end{algorithm}

\section{Numerical results}
\label{sec::numerical_results}

Here, we investigate the performance of a family of quantum LDPC codes implemented in 2D using the 2D local syndrome extraction circuit described in Algorithm~\ref{algo:linear_space_circuit}.
We chose this circuit because it uses a linear number of ancilla qubits which allows us to maintain a non-vanishing ratio between the number of logical qubits and the number of physical qubits, including all the ancillas consumed by the syndrome extraction circuit.

We simulate the circuit described in Algorithm~\ref{algo:linear_space_circuit} with the difference that we perform the readout qubit preparation (step 1) in the same time as the first operation of loop in step 2 and we perform readout qubit measurement (step 8) in the same time as the last operation of the loop 5.
This reduces the overall depth of the syndrome extraction circuit by $4$.

We consider HGP codes constructed from the product of (3, 4)-regular bipartite graphs with girth at least 8 and which have rate $1/25$.
Including the ancilla qubits used for syndrome extraction the number of logical qubits per physical qubit is $1/98$.
The codes are generated using the procedure of Ref.~\cite{grospellier_combining_2020} and for each code length, we pick the best code from a few hundred samples by decoding i.i.d noise using the (Small Set Flip) SSF decoder~\cite{leverrier_quantum_2015}.
For larger code length, we selected fewer samples to keep the runtime of the selection process reasonable ranging from 50 to 1000 code instance samples.

We use the same decoding scheme as Grospellier and Krishna~\cite{grospellier_numerical_2019}, but replace their idealized noise model by circuit noise generated according to the 2D local syndrome extraction circuit from Algorithm~\ref{algo:linear_space_circuit}.
For comparison, we also simulate the performance of this family of codes with noise generated using fully-connected syndrome extraction circuits described in Algorithm~\ref{algo:fully_connected_circuit} which cannot be implemented locally in 2D.

Our simulations use the standard model of {\em circuit noise model} as follows:
\begin{itemize}
\item Each preparation is followed by a random Pauli error acting on the prepared qubit with probability $p$. The error is selected uniformly from the set $\{X, Y, Z\}$.
\item Each waiting qubit is affected by a random Pauli error acting with probability $p$. The error is selected uniformly from the set $\{X, Y, Z\}$.
\item Each $s$-qubit operation is followed by a random Pauli error acting on its support with probability $p$. The error is selected uniformly from the set of non-trivial $s$-qubit Pauli error.
\item The outcome of each measurement is flipped with probability $p$.
\end{itemize}
Given that they can be postponed until the end of the circuit and implemented in software without any physical action on the qubits, we assume that classically-controlled Pauli operations are not noisy.

Following Ref.~\cite{grospellier_combining_2020}, decoding is performed for $T$ consecutive rounds.
After each round of stabilizer extraction, Belief Propagation (BP)~\cite{gallager_low-density_1962, poulin_iterative_2008} is iterated until the syndrome weight reaches a local minimum.
In general this does not return the system to the code space after each round but this is enough to avoid error accumulation.
To determine if a logical error occurred after $T$ rounds of error correction, we use a round of perfect stabilizer measurement and then the BP decoder and the SSF decoder~\cite{leverrier_quantum_2015} are alternated until the SSF decoder converges to a correction which has the right syndrome or if the number of cycles reaches its fixed bound.
At the end of this process, this correction based on a round of perfect measurement is applied to the data qubits.
We say that a failure has occurred by round $T$ if the net effect of all noise and corrections has a non-trivial syndrome or if it is a logical error.

In our simulations, we estimate the logical failure rate by averaging over 10 successive rounds of error correction.
We use of more than one round of error correction to ensure that the logical failure rate per round is close to its limit.
Based on our numerical simulation, we found that ten rounds are sufficient.

Figure~\ref{fig::numerical_results} shows our numerical results.
We contrast the performance of the 2D local syndrome extraction circuits of Algorithm~\ref{algo:linear_space_circuit} with the performance achieved using the bounded-depth circuit of Algorithm~\ref{algo:fully_connected_circuit}. 
We see an enormous gap in performance between these two scenarios, showing that a 2D local hardware is not well suited to the implementation of these codes.
We discuss an alterative implementation of HGP codes in 2D based on long-range gates in Ref.~\cite{tremblay2021layers}.


\end{document}